\documentclass[12pt,a4paper]{article}
\usepackage[latin1]{inputenc}
\usepackage{amsmath}
\usepackage{amsthm}
\usepackage{amsfonts}
\usepackage{amssymb}
\usepackage[pagebackref,colorlinks=true,citecolor=blue,linkcolor=blue,urlcolor=blue]{hyperref}
\usepackage{cleveref}
\usepackage{graphicx}
\usepackage{nicefrac}
\usepackage{fullpage}
\usepackage{bm}
\usepackage{xcolor}
\usepackage[linesnumbered,ruled]{algorithm2e}
\usepackage{tikz}
\usepackage{mathtools}
\usepackage[all]{xy}
\usepackage{authblk}
\usepackage{subcaption} 
\usepackage{tikz}

\usepackage{comment}

\usetikzlibrary{calc,positioning}

\newtheorem{theorem}{Theorem}
\newtheorem{proposition}[theorem]{Proposition}
\newtheorem{lemma}[theorem]{Lemma}
\newtheorem{remark}[theorem]{Remark}
\newtheorem{corollary}[theorem]{Corollary}

\newtheorem{definition}[theorem]{Definition}

\newtheorem{example}[theorem]{Example}

\DeclareMathOperator{\linspan}{Span}

\newcommand{\ef}{\mathbb{F}}

\newcommand{\efq}{\ef_q}
\newcommand{\efqm}{\ef_{q^m}}

\newcommand{\bs}{\boldsymbol}

\newcommand{\rs}{\mathrm{RS}}

\newcommand{\isom}{\stackrel{\sim}{\to}}

\newcommand{\aaa}{1.2cm}   
\newcommand{\bbb}{0.8cm}   

\newcommand{\row}{\mathrm{Row}}
\newcommand{\col}{\mathrm{Col}}

\newcommand{\df}{\mathfrak{n}}

\title{Codes with Large Minimum Distance in Product Codes:
Explicit Constructions and Bounds}    

\author[1]{Amit Berman}
\author[1]{Yaron Shany}
\author[2,1]{Itzhak Tamo}
\affil[1]{Samsung  Semiconductor Israel R\&D Center, 146 Derech
Menachem Begin St., Tel Aviv 6492103, Israel. Emails: \{amit.berman,
yaron.shany\}@samsung.com} 
\affil[2]{Department of Electrical Engineering-Systems, Tel Aviv
University, Tel Aviv 6997801, Israel. Email: zactamo@gmail.com}

\date{\today}

\begin{document}

\maketitle

\begin{abstract}
Products of MDS codes are of major practical importance; for a recent example, they are used in Data Availability Sampling (DAS) in blockchain networks such as Celestia and as part of the Ethereum roadmap. This motivates us to consider subcodes of such codes with the goal of obtaining a larger minimum distance. In this paper, we present explicit constructions of subcodes of Reed--Solomon product codes, along with bounds on their minimum distance. In particular, they achieve an optimal or near-optimal dimension--distance tradeoff. For component codes of dimension $r$, our construction requires a field whose size is bounded linearly by the overall product code length, and attains the maximum possible minimum distance for subcode dimensions $r^2-1$, $r^2-2$, and all dimensions at most $2r-1$. Furthermore, we establish a new upper bound on the minimum distance of subcodes of the product of two codes with identical parameters. 
\end{abstract}

\section{Introduction}\label{sec:intro}

Product codes are a classical and versatile family of error-correcting codes. Introduced by Elias~\cite{Elias54}, they are obtained by arranging information in a two-dimensional array and encoding the rows and columns separately. Their appeal lies in the strong structural properties inherited from this tensor construction: rows and columns remain codewords of the component codes, which supports efficient local repair while also providing nontrivial global distance. These features make product codes attractive in a variety of settings where one seeks to balance local recoverability and global reliability.

A recent and particularly relevant application arises in blockchain systems based on sharding and rollups, where the \emph{data-availability} (DA) problem has become a central challenge. Light clients should be able to verify that the data needed to reconstruct a block has actually been published, without downloading the full block. Data-availability sampling (DAS) addresses this by having clients sample only a small number of encoded symbols rather than the entire block~\cite{AlBassam2018,NazirkhanovaNT22}. In a typical two-dimensional DAS scheme, the data is arranged in an $r\times r$ array and encoded by a product code, typically the tensor product of two Reed--Solomon (RS) codes, producing an $n\times n$ array. In this context, the row and column structure of product codes supports both local repair and probabilistic availability checks. Recent analyses of RS-based DAS protocols also underscore the practical relevance of the underlying rate, distance, and field-size tradeoffs~\cite{paolo}.
However, the main limitation of the ambient product code remains its global minimum distance. If $C_1$ and $C_2$ have minimum distances $d_1$ and $d_2$, then their product $C_1\otimes C_2$ has minimum distance $d_1d_2$; equivalently, there exist erasure patterns of size $d_1d_2$ that are not uniquely recoverable. This motivates the search for subcodes with larger global distance while preserving the same row and column structure. Put differently, we seek to add a small number of global constraints (``heavy parities'') to an RS product code so as to improve the dimension--distance tradeoff without sacrificing the structural features that make product codes attractive for DAS.

The study of subcodes of product codes with enhanced erasure resilience is closely related to the literature on \emph{maximally recoverable} (MR) codes for grid-like topologies~\cite{Gopalan2017,Kong2021,Brakensiek2025}. Grid MR codes are designed to correct all erasure patterns that are information-theoretically recoverable given the row and column constraints. While they offer optimal erasure resilience, explicit constructions of MR grid codes typically require field sizes that grow exponentially with the grid dimensions. In contrast, our focus is on obtaining explicit constructions with guaranteed large global minimum distances over small, practically viable finite fields.

To study this tradeoff, we view the problem through the framework of \emph{locally recoverable codes} (LRCs) with availability. Recall that a code has local recovery if an erased symbol can be reconstructed by accessing only a small subset of other symbols, called a \emph{recovering set}. If each symbol has several disjoint recovering sets, the code is said to have availability. Such codes have been studied extensively from both the bounds and construction perspectives~\cite{Tamo2014,TamoBF16,kathy}.

In our setting, the component codes are \emph{Maximum Distance Separable} (MDS) codes, such as the RS codes considered in this work. An MDS code of length $n$ and dimension $r$ achieves the Singleton bound with equality, namely minimum distance $\delta=n-r+1$. Therefore, in a 2D product code obtained from two $[n,r,\delta]$ MDS codes, each symbol lies in both a row codeword and a column codeword, yielding two disjoint recovering sets. More precisely, each coordinate belongs to two local $[n,r,\delta]$ MDS constraints, one from its row and one from its column. Consequently, a single erasure can be recovered from any $r$ surviving symbols in either of these two sets. In the terminology of locally recoverable codes, this means that the code has all-symbol locality $r$ with local distance $\delta$, or equivalently all-symbol $(r,\delta)$-locality.

Since any such subcode inherits this all-symbol $(r,\delta)$-locality, its minimum distance $d$ satisfies the standard Singleton-type upper bound for LRCs with $(r,\delta)$-locality~\cite{Kamath2014}. For a subcode of dimension $k \le r^2$, this bound is given by
\begin{equation}\label{eq:1dlrc}
d\leq n^2-k+1-\left\lfloor \frac{k-1}{r} \right\rfloor (\delta-1).
\end{equation}

While~\eqref{eq:1dlrc} provides a fundamental baseline, it leaves significant room for improvement because it only accounts for a single local recovery option. Exploiting the fact that product codes provide \emph{two} disjoint recovering sets for each symbol should yield stronger upper bounds on the maximum possible minimum distance. Existing upper bounds for LRCs with multiple recovering sets, however, have limitations when applied to our setting. For the case of $\delta=2$, Wang and Zhang~\cite{WZ14,Wang2014} obtained a sharp upper bound on the minimum distance of subcodes of the product of two $[r+1,r]$ codes, alongside an explicit construction. However, their construction requires a field size that is exponential in the product code length. 

For arbitrary local distances $\delta \ge 3$, more recent works have generalized these bounds~\cite{Cai2020}. Yet, because these bounds are derived for general LRCs, they do not fully capture the rigid intersecting grid topology of product codes and remain notably loose in the high-rate regime. In light of the sharp bounds already existing for $\delta=2$, and the remaining gap for higher local distances, this paper restricts its attention to the case $\delta \ge 3$.

\paragraph{Our Contributions.}
In this paper, we address both the theoretical bounds and the explicit construction of such codes. Our main contributions are as follows:
\begin{itemize}
    \item \textbf{A New Upper Bound:} We establish a new upper bound on the minimum distance of subcodes of the product of two codes with identical parameters. In contrast to existing general bounds for LRCs with two recovering sets per symbol~\cite{Wang2014,Cai2020}, our bound explicitly captures the grid structure, and is valid for an arbitrary local distance $\delta$.
    
    \item \textbf{Explicit Constructions over Small Fields:} We present a new explicit family of subcodes of RS product codes. A key feature of our construction is a mathematical technique that maps bivariate polynomial evaluations to univariate polynomials evaluated over the direct sum of roots of linearized polynomials. This allows the construction to operate over a finite field whose size is bounded linearly by the overall product code length (e.g., requiring a field of size at least  $ n^2$ for an $n \times n$ product code), which is highly desirable for practical implementations.
    
    \item \textbf{Strict Optimality:} We prove that our constructed subcodes achieve an optimal or near-optimal dimension--distance tradeoff. Specifically, for component codes of dimension $r$, our codes attain the maximum possible minimum distance for subcode dimensions $r^2-1$ (one heavy parity) and $r^2-2$ (two heavy parities), exactly matching our new upper bound. Furthermore, for all dimensions at most $ 2r-1$, they perfectly meet the Kamath \textit{et al.}\ bound~\eqref{eq:1dlrc}.
\end{itemize}

\paragraph{Paper Organization.}
The remainder of the paper is organized as follows. Section~\ref{sec:prelim} provides preliminaries and notation. Section~\ref{sec:codes} introduces our explicit univariate construction of subcodes of RS product codes. In Section~\ref{sec:bound}, we derive our new general upper bound on the minimum distance. Section~\ref{sec:5} proves the minimum distance lower bounds for our construction and establishes its strict optimality for specific dimensions.  Sections~\ref{sec:inst} and~\ref{sec:eg} provide concrete instantiations and numerical examples comparing our bounds. Finally, in Section~\ref{sec:conc}, we conclude and discuss open questions.

\section{Preliminaries}\label{sec:prelim}
This section includes some definitions and notation used throughout
the paper.

For a prime power $q$, let $\efq$ be the finite field of $q$
elements. Unless otherwise noted, all vectors are row vectors. We use
$(\cdot)^T$ for matrix transposition. 
For a matrix $x\in \efq^{n_1\times n_2}$ (where $n_1,n_2$ are positive integers), we write $\row_i(x)\in
\efq^{1\times n_2}$ for the $i$-th row of $x$, and $\col_j(x)\in
\efq^{n_1\times 1}$ for the $j$-th column of $x$, where
$i\in\{1,\ldots,n_1\}$, $j\in\{1,\ldots,n_2\}$. 
For a polynomial $f\in K[x]$ (for some field $K$), we write $Z_f$ for
the set of zeros (i.e., roots) of $f$ in its splitting field. For
finite sets $X,Y$, we write $Y^X$ for the set of all functions 
$X\to Y$. Finally, we use the notation $[n,k,d]_q$  for an
$\efq$-linear code with length $n$, dimension $k$, and minimum
distance $d$.\footnote{Throughout, ``distance'' and ``weight'' refer to
Hamming distance and Hamming weight, respectively.} When $q$ is clear
from the context, we omit it and write simply $[n,k,d]$ for short. If $F$ is a field and $S\subseteq F$ is a finite subset, the {\bf annihilator polynomial} of $S$ is $\prod_{\alpha\in S} (x-\alpha)\in F[x]$.

\subsection{Tensor product of codes}\label{sec:prod}
For two linear codes $C_1\subseteq \efq^{n_1}, C_2\subseteq \efq^{n_2}$
(where $n_1,n_2$ are positive integers), we let $C_1\otimes
C_2\subseteq \efq^{n_1\times n_2}$ be the set of all matrices $c$
with $\col_j(c)^T\in C_1$ and $\row_i(c)\in C_2$ for all $i,j$. A
different, yet equivalent definition, is that $C_1\otimes C_2$ is the
subspace of $\efq^{n_1\times n_2}$ generated by all {\bf
simple tensors}, $\bs{c}_1\otimes \bs{c}_2:=\bs{c}_1^T\bs{c}_2$, where
$\bs{c}_1\in C_1$, $\bs{c}_2\in C_2$. The code $C_1\otimes C_2$ is
called the {\bf (tensor) product code} of $C_1$ and $C_2$.

It is well known that if $d_1,d_2$ are the minimum
distances of
$C_1,C_2$ (resp.), then the minimum distance of $C_1\otimes C_2$ is
$d_1d_2$. Also, writing $k_i:=\dim(C_i)$ for $i=1,2$, if
$B^{(1)}=\{\bs{b}^{(1)}_1,\ldots,\bs{b}^{(1)}_{k_1}\}$ and
$B^{(2)}=\{\bs{b}^{(2)}_1,\ldots,\bs{b}^{(2)}_{k_2}\}$ are bases for
$C_1,C_2$ (resp.), then $\{\bs{b}^{(1)}_i\otimes
\bs{b}^{(2)}_j\}_{i,j}$ is a basis for $C_1\otimes C_2$. Consequently,
$\dim(C_1\otimes C_2)=\dim(C_1)\dim(C_2)$.

\subsection{Reed--Solomon codes}
For integers $k\geq 0$ and $1\leq n\leq q$ with $n\geq k$, and for a
vector 
$\bs{a}=(a_1,\ldots,a_n)$ of distinct 
elements from $\efq$, the {\bf Reed--Solomon (RS) code} $\rs_q(k,\bs{a})$ of
dimension $k$ and {\bf evaluation vector} $\bs{a}$ is defined as
\begin{equation}\label{eq:RS}
\rs_q(k,\bs{a}):=\big\{(f(a_1),f(a_2),\ldots,f(a_n))\big|f\in \efq[x],
\deg(f)<k\big\}.
\end{equation}
When the underlying finite field is clear from the context, we omit the
subscript $q$, and write simply $\rs(k,\bs{a})$ for
$\rs_q(k,\bs{a})$. It is easily verified that $\rs_q(k,\bs{a})$ has
dimension $k$ and minimum distance $n-k+1$, and it is therefore an MDS
code.

With the above notation, let $A:=\{a_1,\ldots,a_n\}$. It is sometimes
more convenient to avoid element indexing, and to slightly modify the
definition of \eqref{eq:RS} as
$$
\rs_q(k,A):=\big\{f\colon A\to \efq\big|f\text{ is a polynomial
function over $\efq$ of degree }<k\big\}.
$$

The tensor product of two RS
codes can be described as all evaluation matrices of an
appropriate set of bivariate polynomials. In detail, letting also
$\bs{b}=(b_1,\ldots,b_{n'})$ be a vector of distinct elements from $\efq$
(where $n'>0$ is an integer) and $0\leq k'\leq n'$ be an integer, 
it follows from the discussion at the end of Section \ref{sec:prod}
that the following is a basis for $\rs_q(k,\bs{a})\otimes 
\rs_q(k',\bs{b})$:
\begin{equation}\label{eq:prodbasis}
\Big\{(a_i^s b_j^t)_{1\leq i\leq n, 1\leq j\leq n'}|0\leq s<k, 0\leq t<k'\Big\}. 
\end{equation}
Consequently, 
\begin{equation}\label{eq:rsprod}
\rs_q(k,\bs{a})\otimes \rs_q(k',\bs{b})=\Big\{f(a_i,b_j)_{1\leq i\leq
n, 1\leq j\leq n'}|f\in \efq[x,y],
\deg_x(f)<k, \deg_y(f)<k'\Big\}.
\end{equation}

Considering \eqref{eq:rsprod}, it is useful to have the following
definition.

\begin{definition}\label{def:polys}
For non-negative integers $r_1,r_2$ and for a field $F$, let
$$
F[x,y]^{(r_1,r_2)}:=\{f\in F[x,y]|\deg_x(f)<r_1, \deg_y(f)<r_2\}.
$$
\end{definition}

\subsection{Linearized polynomials}
For a non-negative integer $n$, a polynomial of the form
$f(x)=a_nx^{q^n}+a_{n-1}x^{q^{n-1}}+\cdots +a_0 x\in \efqm[x]$ (where $m$ is
a positive integer) is called a {\bf $q$-linearized polynomial} over
$\efqm$. Such a polynomial $f$ clearly defines an $\efq$-linear map, and
hence $Z_f$ is an $\efq$-vector space of dimension at most $n$. Also,
if $F$ is a finite extension of $\efq$ and  
$V\subseteq F$ is an $\efq$-vector space, then the annihilator polynomial of $V$  is a $q$-linearized polynomial over $F$ (see, e.g.,
\cite[Thm.~3.52]{LN}). Finally, 
the formal derivative $f'$ of $f$ equals $a_0$, and therefore $f$ is
separable iff $a_0\neq 0$. Hence, when $a_0\neq 0$, $\dim(Z_f)=n$. 

\section{An explicit construction}\label{sec:codes}
In this section, we present an explicit family of subcodes of the
product of two RS codes.  A key feature is that, as opposed to the
description in \eqref{eq:rsprod}, codewords of the product code are
described as evaluation vectors of \emph{univariate}
polynomials. Following a careful degree analysis, we define the
subcodes by limiting the degree of the evaluated univariate
polynomials. 

Throughout, we fix a prime power $q$. Let $f(x)=\sum_{i}a_ix^{q^i}$ be
a $q$-linearized polynomial over some finite extension $K$ of $\efq$,
and assume that $a_0\neq 0,1$ (this requires $|K|>2$) and that
$\deg(f)>1$. Define $g:=x-f\in K[x]$. It follows from the restriction on $a_0$ that both $f$ and $g$ are separable $q$-linearized polynomials. Let $F=\efqm$ be the splitting field of
$\{f,g\}$ (this defines $m$). By definition,
for any non-zero $\beta\in Z_f$, we have
$g(\beta)=\beta-f(\beta)=\beta \neq 0$. Hence
$Z_f\cap Z_g=\{0\}$, and the sum $Z_f+Z_g=Z_f\oplus Z_g\subseteq F$
is direct. We note that in Section \ref{sec:inst} below, we will
describe some simple concrete examples with $Z_f\oplus
Z_g=F=\ef_{q^2}$. 

As already mentioned, central to the definition and distance analysis of the codes constructed below is the presentation of certain product codes as evaluation codes of univariate polynomial spaces. The following proposition is the main step in this direction, where the commutative diagram in the proposition serves as a sort of dictionary for moving between the familiar bivariate presentation and the more useful univariate presentation.

From this point on, we let $\df:=\deg(f)=\deg(g)$ (this is a power of $q$), and  let $r\leq \df$ be a non-negative integer. Define $B_{r^2}:=\{g^if^j|0\leq i,j\leq
r-1\}$.

\begin{proposition} \label{prop:diagram}
Let $\psi\colon Z_f\oplus Z_g\isom Z_f\times Z_g$ be the
bijection $\beta+\gamma\mapsto (\beta,\gamma)$. Then
the 
following diagram of $F$-vector spaces and $F$-linear maps is commutative.
\begin{equation}\label{eq:diagram}
\xymatrix{
F[x,y]^{(r,r)} \ar[rr]^(0.4){x^iy^j\mapsto g^if^j}\ar[dd]_{\mathrm{eval}} &&
\linspan_{F}(B_{r^2})\subset F[x]\ar[dd]_{\mathrm{eval}}\\
\\
{F}^{Z_f\times Z_g} \ar[rr]_{\sim}^{(c\colon Z_f\times Z_g\to F)\mapsto
c\circ\psi} && {F}^{Z_f\oplus Z_g}
}
\end{equation}
Here, in the top horizontal map, the linear map is defined by its
values on a basis, and the vertical maps are evaluation maps taking
a polynomial $u(x,y)$ (resp., $v(x)$) to its evaluation vector
$(\beta,\gamma)\mapsto u(\beta,\gamma)$ (resp., $\alpha\mapsto v(\alpha)$).
\end{proposition} 

\begin{proof}
We note that for $(\beta,\gamma)\in Z_f\times Z_g$, it holds that
$g(\beta+\gamma)=g(\beta)$, since $g$ is linearized and $\gamma\in Z_g$. Also, from $f+g=x$, we have
$g(\beta)=(f+g)(\beta)=\beta$. Hence
$g(\beta+\gamma)=\beta$. Similarly, $f(\beta+\gamma)=\gamma$.

Now commutativity can be verified directly by
chasing a basis element at the top left corner, as follows:
$$
\xymatrix{
x^iy^j\ar@{|->}[r]\ar@{|->}[d] & g^if^j \ar@{|->}[d] \\
((\beta,\gamma)\mapsto \beta^i\gamma^j)\ar@{|->}[r] &
(\beta+\gamma\mapsto \beta^i\gamma^j)
}
$$
\end{proof}

\begin{remark}\label{rem:relabel}
{\rm
\begin{enumerate}

\item The top horizontal map in \eqref{eq:diagram} can be written
explicitly as $s(x,y)\mapsto s(g(x),f(x))$, for $s(x,y)\in F[x,y]^{(r,r)}$. 

\item The bottom horizontal map in \eqref{eq:diagram}
is just a re-labeling of the coordinates of vectors:\footnote{Here, a matrix is considered as a vector indexed by the Cartesian product of two finite sets.} we re-label
coordinate  $(\beta,\gamma)$ by $\beta+\gamma$. 

\item Recalling \eqref{eq:rsprod}, the image of the left vertical map
in \eqref{eq:diagram} is $\rs_{q^m}(r,Z_f)\otimes\rs_{q^m}(r,Z_g)$.

\item The top horizontal map in \eqref{eq:diagram} is evidently
surjective, as it hits $B_{r^2}$ by definition.

\end{enumerate}
}
\end{remark}

\begin{corollary}\label{coro:codedef}
The set $B_{r^2}$ is linearly independent over $F$, and the evaluation
code of $\linspan_F(B_{r^2})$ on $Z_f\oplus Z_g$ (image of right vertical
map in \eqref{eq:diagram}) is the 
product code $\rs_{q^m}(r,Z_f)\otimes\rs_{q^m}(r,Z_g)$, up to a
relabeling of the coordinates. Also, the top horizontal map in
\eqref{eq:diagram} is an isomorphism.
\end{corollary}

\begin{proof}
As already mentioned, the image of the left vertical map is
$\rs_{q^m}(r,Z_f)\otimes\rs_{q^m}(r,Z_g)$, and the bottom horizontal
map is just a re-labeling of the coordinates. Hence, using the
commutativity of the diagram, the image of the composition of the top horizontal and right vertical maps
is $\rs_{q^m}(r,Z_f)\otimes\rs_{q^m}(r,Z_g)$, up to a relabeling of
the coordinates. This is also the image of the right vertical map,
by surjectivity of the top horizontal map. Since this image is the
product of RS codes, it has
dimension $r^2=|B_{r^2}|\geq \dim(\linspan_F(B_{r^2}))$. As the dimension of the image cannot be larger than that of the domain, it follows that
$\dim(\linspan_F(B_{r^2}))=r^2$, and $B_{r^2}$ is linearly
independent. Finally, the top horizontal arrow is an isomorphism,
because it takes a basis to a basis. 
\end{proof}

Since the codes we construct are subcodes of the evaluation code from
Corollary \ref{coro:codedef}, the following definition will be useful.

\begin{definition}\label{def:uppercode}
Let $C_{r^2}$ be the image of the right vertical map of
\eqref{eq:diagram}. In other words, $C_{r^2}$ is the evaluation code
of $\linspan_F(B_{r^2})$ on $Z_f\oplus Z_g$. By Corollary
\ref{coro:codedef}, this is effectively
$\rs_{q^m}(r,Z_f)\otimes\rs_{q^m}(r,Z_g)$.   
\end{definition}
We note that $C_{r^2}$ has length
$\df^2$, dimension $r^2$, and minimum distance $(\df-r+1)^2$.

In $B_{r^2}$, there are distinct polynomials with identical
degrees. To continue, we would like to specify the degrees of a basis
for $\linspan_F(B_{r^2})$ consisting of polynomials of distinct
degrees. Such a basis exists by linear algebra: if we replace
polynomials of $B_{r^2}$ by their vectors of coefficients in
descending degree order and take the matrix with these vectors as
rows, we are looking for a row echelon form (REF). We note that while
the REF is not unique, the resulting set of distinct degrees
\emph{is} unique: these are exactly the distinct degrees in
$\linspan_F(B_{r^2})$. The degrees in such a
basis are described in Theorem \ref{thm:max-deg-basis}. 
\begin{theorem}[Degrees for a basis in REF]
\label{thm:max-deg-basis}
Let $D$ be the set of degrees of the polynomials for a basis of
$\linspan_F(B_{r^2})$ in REF. Then
\begin{equation}\label{eq:d1d2}
D =\Big\{t\df+\ell \Big|0\leq t\leq 2(r-1),\ \ 0\leq \ell\leq
r-1-\left\lceil \frac{t}{2} \right\rceil\Big\}.
\end{equation}
\end{theorem}

For the proof, see Appendix \ref{app:proof}.

\begin{remark}\label{rem:modulo}
{\rm
The largest number in $D$ is $2(r-1)\df$, which is at least $\df^2$
if and only if $r\geq \df/2+1$.  Let $a(x)=\prod_{\alpha\in Z_f\oplus Z_g} 
(x-\alpha)$ be the annihilator polynomial of $Z_f\oplus Z_g$. Clearly,
the evaluation vectors on $Z_f\oplus Z_g$ of $h(x)\in F[x]$ and
$h(x)\bmod a(x)$ are the same. Hence, it is possible to replace the space
$\linspan_{F}(B_{r^2})$ in \eqref{eq:diagram} by the space of
polynomials reduced modulo $a(x)$ (and also replace the top horizontal
map to $s(x,y)\mapsto s(g,f)\bmod a(x)$). This will typically result
in different codes $C_k$ in Definition \ref{def:ck} below, and may perhaps
improve the lower bounds on the minimum distance derived below when
$r\geq \df/2+1$, but it seems 
complicated to extend the degree analysis of Theorem 
\ref{thm:max-deg-basis} to this setup.  We leave this as an open
question for future research.  
}
\end{remark}

We can now finally define the family of subcodes of 
$C_{r^2}=\rs_{q^m}(r,Z_f)\otimes\rs_{q^m}(r,Z_g)$:
\begin{definition}[The codes $C_k$]\label{def:ck}
Enumerate the set $D$ of Theorem \ref{thm:max-deg-basis} as 
$D=\{\partial_1,\partial_2,\ldots,\partial_{r^2}\}$ where
$\partial_1<\partial_2<\cdots<\partial_{r^2}$. For $k=1,\ldots,r^2$,
we define
$$
C_k:=\text{evaluation code of }\{w(x)\in \linspan_F(B_{r^2})|\deg(w)\leq
\partial_k\}\text{ on }Z_f\oplus Z_g,
$$
and let $d_k$ be the minimum distance of $C_k$. Note that by
construction, $\dim(C_k)=k$ for all $k$; a basis for 
$C_k$ can be taken as the set of evaluation vectors of the $k$ lowest
degree polynomials in a basis for $\linspan_F(B_{r^2})$ in REF. The
number $h$ of {\bf heavy parities} of $C_k$ is defined as the
difference $h:=r^2-k$.
\end{definition}
Note that for $k=r^2$, the code $C_k$ of Definition \ref{def:ck} is
indeed the entire code $C_{r^2}$ of Definition \ref{def:uppercode}.

For finding lower bounds on the minimum distance of the codes, it will
be useful to have an explicit description of the degrees
$\partial_k$.

\begin{proposition}\label{prop:degrees}
For integer $0\leq t \leq 2r-2$, let
\begin{equation}\label{eq:kt}
k_t=(t+1)r-\left\lfloor \frac{t+1}{2}\right \rfloor \cdot
\left\lceil\frac{t+1}{2} \right\rceil.
\end{equation} 
Then 
\begin{equation}\label{eq:degreekt}
\partial_{k_t}=t\df+r-1-\left\lfloor
\frac{t+1}{2}\right\rfloor.
\end{equation}
Also, with the convention that $k_{-1}:=0$, for $0\leq t \leq 2r-2$
and for $k_{t-1}<k<k_t$,
$\partial_{k}=\partial_{k_t}-(k_t-k)$.
\end{proposition}

\begin{proof}
It follows from Theorem \ref{thm:max-deg-basis} that the set $D$ of
degrees is the union of the intervals $I_t=\{t\df+\ell |0\leq \ell\leq
r-1-\lceil t/2\rceil\}$ defined in Lemma \ref{lemma:sum} of Appendix \ref{app:proof}. By this
lemma, $k_t$ of \eqref{eq:kt} equals $\sum_{j=0}^t |I_j|$, and clearly
$\partial_{k_t}=\max I_t$, which is \eqref{eq:degreekt}. This proves
the first assertion, 
and the second assertion follows since each $I_j$ forms a contiguous
interval of integers. 
\end{proof}

\begin{remark}\label{rem:degrees}
{\rm
It is natural to try to bound the minimum distance of $C_k$ using only
the univariate description from Definition~\ref{def:ck}. Namely, if a
nonzero codeword of weight $w$ is obtained by evaluating a polynomial $u(x)$ of
degree at most $\partial_k$ on the set $Z_f\oplus Z_g$ of size
$\df^2$, then the RS-type degree lower bound\footnote{Here and throughout, the ``RS-type degree lower
bound'' on the weight $w$ of the evaluation vector of a polynomial
$u(x)$ on a set of size $n$ is $w\geq n-\deg(u)$.} gives
\[
w\ge \df^2-\deg(u)\ge \df^2-\partial_k.
\]
The endpoint $k=r^2$ shows that this estimate is too weak for our
purposes. Indeed, for $t=2(r-1)$ we have $k_t=r^2$ and
$
\partial_{r^2}=2(r-1)\df$, 
so the above argument yields
\[
d_{r^2}\ge \df^2-\partial_{r^2}
=\df^2-2(r-1)\df
=(\df-r+1)^2-(r-1)^2.
\]
However, $C_{r^2}$ is the full product code
$\rs_{q^m}(r,Z_f)\otimes\rs_{q^m}(r,Z_g)$, whose true minimum distance
is $(\df-r+1)^2$. 
Thus, the naive univariate degree bound misses the correct value by
$(r-1)^2$. In particular, it does not even recover the minimum
distance of the ambient product code. 

In Section~\ref{sec:5}, we derive an improved lower bound on the
minimum distance that uses not only the degree restriction
$\deg(u)\le \partial_k$, but also the additional structure coming from
the tensor-product description and the commutative diagram
\eqref{eq:diagram}. This improved bound is always at least the minimum
distance of the ambient product code, and for $k=r^2-1$ and
$k=r^2-2$ it matches the upper bound of Proposition~\ref{prop:bound} below.}
\end{remark}

\section{An upper bound}
\label{sec:bound}
Before deriving lower bounds on the minimum distance of the codes $C_k$, we first present a general upper bound. This bound will serve to quantify the gap between the minimum distance of the codes $C_k$ and the maximum attainable value, and to show that this maximum is achieved for $k = r^2 - 1, r^2 - 2$. As highlighted in the introduction, deriving this tailored upper bound is necessary because existing general bounds do not tightly capture the rigid grid topology of product codes for local distances $\delta \ge 3$.

\begin{proposition}\label{prop:bound}
Let $C_1,C_2$ be two linear codes of length
$n$ and minimum distance at least $n-r+1$ (for some  integers $1\leq
r\leq n$), and let 
$C\subseteq C_1\otimes C_2$ be a linear subcode. Write
$k:=\dim(C)\leq r^2$ and let $d$ be the minimum distance of $C$. 
Then for any
non-negative integers $a,b\leq r$ it holds that 
\begin{equation}\label{eq:gen_bound}
ab\geq r^2-k+1\quad \implies\quad d\leq (a+n-r)(b+n-r).
\end{equation}
In particular,
\begin{equation}\label{eq:bound}
d\leq \Big(\Big\lceil \sqrt{r^2-k+1}\Big\rceil+n-r
\Big)^2. 
\end{equation}

\end{proposition}

\begin{proof}
\begin{figure}[h!tbp]
\centering
\begin{tikzpicture}[ 
    every node/.style={font=\footnotesize\sffamily}, 
    thin/.style={line width=0.35pt},
  ]


  \draw[line width=0.9pt] (0,0) rectangle (5*\aaa,5*\aaa);

  \fill[white] (0,\aaa) rectangle (4*\aaa,3.4*\aaa);       
  \fill[white] (2*\aaa,3.4*\aaa) rectangle (4*\aaa,5*\aaa); 

  \fill[gray!65] (2*\aaa,0) rectangle (4*\aaa,\aaa);         
  \fill[gray!65] (4*\aaa,\aaa) rectangle (5*\aaa,3.4*\aaa);     

  \fill[black] (0,3.4*\aaa) rectangle (2*\aaa,5*\aaa);         
  \fill[black] (4*\aaa,3.4*\aaa) rectangle (5*\aaa,5*\aaa);     
  \fill[black] (0,0) rectangle (2*\aaa,\aaa);                 
  \fill[black] (4*\aaa,0) rectangle (5*\aaa,\aaa);            

  \node[left=13mm of {(\aaa,4.2*\aaa)}]   {$b$};
  \node[left=13mm of {(\aaa,2.2*\aaa)}]   {$r-b$};
  \node[left=13mm of {(\aaa,0.5*\aaa)}] {$n-r$};

  \node[above=1mm of {(1*\aaa,5*\aaa)}]     {$a$};      
  \node[above=1mm of {(3*\aaa,5*\aaa)}]     {$r-a$};    
  \node[above=1mm of {(4.5*\aaa,5*\aaa)}]   {$n-r$};    

\end{tikzpicture}
\caption{The erasure pattern for the proof of Proposition \ref{prop:bound}.}
\label{fig:pattern}
\end{figure}
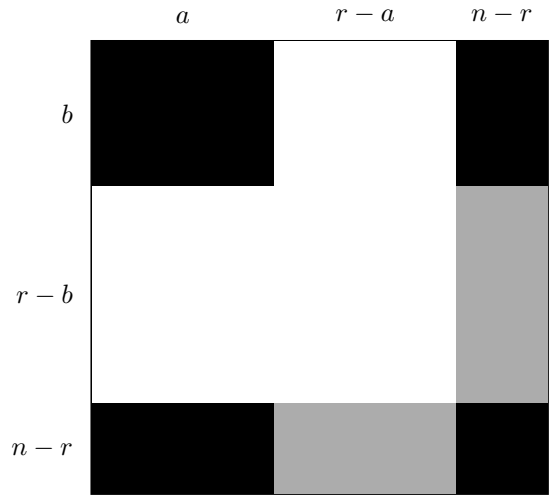

In Figure \ref{fig:pattern}, let $a,b\leq r$ be some non-negative
integers to be determined later.  If the black and gray parts are erased,
then local decoding can be used to restore the gray part,\footnote{Specifically, the bottom-middle gray rectangle has height $n-r$, and since the white regions above it are not erased, these columns contain exactly $n-r$ erasures. Because the local column codes have minimum distance at least $n-r+1$, they can  correct these erasures. The same logic applies to the middle-right gray rectangle using the local row codes.} and then if
the size of the black part is at most $d-1$, then the overall erasure
pattern can be decoded. If the total size of the black and gray parts
is strictly larger than the codimension $n^2-k$ of $C$, this leads to
 contradiction, because an $[n^2, k]$ linear code cannot uniquely recover from any erasure pattern of size $\ge n^2 - k + 1$. In other words, if 
\begin{equation}\label{eq:aval}
\underbrace{(a+n-r)(b+n-r)}_{\text{black}} +
\underbrace{(n-r)(r-a+r-b)}_{\text{gray}} \geq n^2-k+1,
\end{equation}
for some non-negative integers $a,b\leq r$, then $(a+n-r)(b+n-r)\geq d$. Put
differently, 
$d\leq (a+n-r)(b+n-r)$ for any non-negative integers $a,b\leq r$ satisfying
\eqref{eq:aval}.
It is easily verified that \eqref{eq:aval} is equivalent to $ab\geq
r^2-k+1$, which completes the proof of the first assertion. The second
assertion follows by letting $a=b=\lceil\sqrt{r^2-k+1} \rceil$, which
is clearly $\leq r$ for $1\leq k\leq r^2$.
\end{proof}

To demonstrate that the bound can be tight, we begin by considering the endpoint values $k=1,r^2$.
\begin{example}[Tightness of the bound for $k=1,r^2$]
{\rm
Suppose that $C_1,C_2$ are MDS. For $k=r^2$, the bound
\eqref{eq:bound} reads $d\leq 
(n-r+1)^2$, and the right-hand side is indeed the minimum distance of
the product code. Also, for $k=1$, \eqref{eq:bound} reads $d\leq n^2$,
which is the minimum distance of the one-dimensional subcode of
$C_1\otimes C_2$ generated by the simple tensor of two vectors of
weight $n$ from $C_1$ and $C_2$ (such vectors are guaranteed for MDS 
codes).\footnote{For the  case of RS codes, this is
obtained by evaluating a non-zero constant polynomial, and the general
case follows, because MDS codes of the same length, dimension, and
underlying field, have the same weight distribution.} We conclude that
the bound is tight at the endpoints of $k=1,r^2$.
}
\end{example}

Next, we would like to show that when $C_1,C_2$ are MDS and $k=r^2-1$,
the bound of Proposition \ref{prop:bound} is the next-to-minimum
weight of $C_1\otimes C_2$, that is, the smallest weight that is
strictly larger than the minimum distance. For this, we first have the
following proposition. 

\begin{proposition}\label{prop:secondw}
Let $C_1,C_2$ be linear codes of minimum distance $\delta$, where at
least one of $C_1$, $C_2$ has a word of weight $\delta+1$. Then the
next-to-minimum weight of $C_1\otimes C_2$ is $\delta(\delta+1)$.
\end{proposition}

\begin{proof}
Suppose w.l.o.g.~that $C_1$ has a word $\bs{c}_1$ of weight
$\delta+1$, and let $\bs{c}_2\in C_2$ have weight $\delta$. Then
$\bs{c}_1\otimes \bs{c}_2$ has weight $\delta(\delta+1)$, so that the
next-to-minimum weight of $C_1\otimes C_2$ is at most
$\delta(\delta+1)$. Conversely, let
$n_r,n_c$ be the numbers of non-zero rows and columns (resp.) in 
some non-zero word $\bs{c}$ of $C_1\otimes C_2$. 
To have a weight strictly larger than $\delta^2$, at least one of $n_r,n_c$
must be at least $\delta+1$, in which case the weight of $\bs{c}$ is at
least $\delta(\delta+1)$.
\end{proof}

\begin{example}\label{example:hone}
{\rm
If  $k=r^2-1$, we have  $r^2-k+1=2$, 
and we may choose $(a,b)=(1,2)$ in \eqref{eq:gen_bound}, resulting in
$d\leq (n-r+1)(n-r+2)$. By Proposition \ref{prop:secondw}, this is the
next-to-minimum weight of $C_1\otimes C_2$ when $C_1,C_2$ are MDS. We will see
below in Corollary \ref{coro:hone} that the code $C_k$ with $k=r^2-1$ achieves this bound, and it
is therefore tight. This example also demonstrates that in some cases,
a tighter bound than \eqref{eq:bound} is obtained by choosing unequal
$a,b$ in \eqref{eq:gen_bound}.
}
\end{example}

For completeness, in Appendix \ref{app:boundv2} we provide an alternative upper bound. While generally looser than Proposition \ref{prop:bound}, it can be slightly sharper in low dimensions; concrete examples illustrating this are provided in Section \ref{sec:eg}.

\section{Minimum distance lower bounds for the codes $C_k$}
\label{sec:5}
Throughout this section, we maintain the assumptions and notation of Section \ref{sec:codes}. 
Recalling Remark \ref {rem:degrees}, our first goal is to derive a lower bound that is tighter than the RS-type degree lower bound for values of $k$ close to $r^2$. We begin with the most general lower bound in Theorem \ref{thm:boundopt}. Although this bound is not presented as an explicit function, we prove that it uniformly outperforms the naive RS-type degree bound and ensures the guaranteed distance never falls below that of the baseline product code. 
We then proceed to finding an explicit expression for the bound  for one and two heavy parities ($k=r^2-1$ and $k=r^2-2$) in Corollaries \ref{coro:hone}, \ref{coro:htwo}, respectively. The lower bound in these two cases perfectly matches the upper bound established in Section \ref{sec:bound},  which proves that our construction is strictly optimal in these cases. It is then proved in Proposition \ref{prop:smallk} that for $k\leq2r-1$, the RS-type degree bound attains the LRC upper bound \eqref{eq:1dlrc}, and hence the construction is optimal also in this case. Finally, it is shown in Corollary \ref{cor:ninety-percent} that the lower bound is a large fraction of the upper bound \eqref{eq:bound} for a broad range of dimensions $k$ when the local codes have rate $\leq 1/2$.

Theorem \ref{thm:boundopt} below heavily uses the structure of the codes $C_k$ and the commutative diagram \eqref{eq:diagram}, moving freely between the bivariate and univariate evaluated polynomials, and building on properties obtained from each representation. A key observation is that at a coordinate that is in the intersection of a zero column and a zero row of a codeword, the evaluated univariate polynomial has a zero \emph{of multiplicity at least $2$}.

\begin{theorem}[Main lower bound]\label{thm:boundopt}
For a non-negative integer $\partial$ and positive integers $n_r,n_c$, let
$$
B(\partial,n_r,n_c):=\max\Big\{\df^2-\partial+(\df-n_r)(\df-n_c),
n_r\delta,n_c\delta \Big\}, 
$$
where $\delta:=\df-r+1$. Then for $1\leq k\leq r^2$, the minimum
distance $d_k$ of $C_k$ satisfies
\begin{equation}\label{eq:boundopt}
d_k\geq \min_{\delta\leq n_r,n_c\leq \df} B(\partial_k,n_r,n_c).
\end{equation}
\end{theorem}

\begin{proof}
Let $c\in C_k$ be a non-zero codeword of weight $w$, and let
$c'\in F^{Z_f\times Z_g}$ be the corresponding word in the usual tensor
product code. Let $n_r,n_c$ be the numbers of non-zero rows and
columns (resp.) in $c'$, and write
$$
X:=\df-n_r,\qquad Y:=\df-n_c,
$$
for the numbers of zero rows and zero columns (resp.). Let
$s(x,y)\in F[x,y]^{(r,r)}$ be the bivariate polynomial whose
evaluation matrix on $Z_f\times Z_g$ is $c'$, and let
$$
h(x):=s(g(x),f(x))\in F[x].
$$
Then by the commutativity of the diagram \eqref{eq:diagram}, $h$ is the univariate polynomial whose evaluation on
$Z_f\oplus Z_g$ gives the codeword $c$, and $\deg(h)\leq \partial_k$ by
definition of $C_k$.

We first prove the lower bound
\begin{equation}\label{eq:firstB}
w\geq \df^2-\partial_k+(\df-n_r)(\df-n_c).
\end{equation}

Let $R\subseteq Z_f$ be the set of row indices of zero rows of $c'$, so
that $|R|=X$, and let $T\subseteq Z_g$ be the set of column indices of
zero columns of $c'$, so that $|T|=Y$.

Now let $\beta\in R$. Since the $\beta$-th row of $c'$ is zero, the
polynomial $s(\beta,y)$ vanishes on all of $Z_g$. But
$\deg_y(s)<r\leq \df=|Z_g|$, hence $s(\beta,y)$ is the zero polynomial.
Therefore $x-\beta$ divides $s(x,y)$. Similarly, if $\gamma\in T$, then
the $\gamma$-th column of $c'$ is zero, and the polynomial $s(x,\gamma)$
vanishes on all of $Z_f$. Since $\deg_x(s)<r\leq \df=|Z_f|$, it follows
that $y-\gamma$ divides $s(x,y)$.

Fix now $\beta\in R$ and $\gamma\in T$. From the above, both
$x-\beta$ and $y-\gamma$ divide $s(x,y)$, and therefore
$$
(g(x)-\beta)(f(x)-\gamma)\mid h(x)=s(g(x),f(x)).
$$
Let
$
\alpha:=\beta+\gamma\in Z_f\oplus Z_g.
$
By Proposition~\ref{prop:diagram},
$
g(\alpha)=g(\beta+\gamma)=\beta,
f(\alpha)=f(\beta+\gamma)=\gamma.
$
Hence both $g(x)-\beta$ and $f(x)-\gamma$ vanish at $\alpha$, so each
is divisible by $x-\alpha$. Consequently,
$$
(x-\alpha)^2\mid h(x).
$$
Thus, for every pair $(\beta,\gamma)\in R\times T$, the point
$\alpha=\beta+\gamma$ is a root of multiplicity at least $2$ of $h$.

We next count root multiplicities of $h$ at the evaluation points
$Z_f\oplus Z_g$. Since the map
$$
\psi^{-1}\colon Z_f\times Z_g\to Z_f\oplus Z_g,\qquad
(\beta,\gamma)\mapsto \beta+\gamma
$$
is a bijection, distinct pairs $(\beta,\gamma)$ correspond to distinct
evaluation points. Every zero entry of $c'$ gives a root of $h$ of
multiplicity at least $1$, and for the $XY$ zero entries in the
rectangle $R\times T$ we have multiplicity at least $2$. Since $c'$ has
$\df^2-w$ zero entries in total, it follows that the sum of the
multiplicities of the roots of $h$ in $Z_f\oplus Z_g$ is at least
$$
(\df^2-w)+XY.
$$
As $h$ is a non-zero univariate polynomial over $F$, the sum of the
multiplicities of all its roots is at most $\deg(h)$. Therefore
$$
\partial_k\geq\deg(h)\geq \df^2-w+XY.
$$
Equivalently,
$$
w\geq \df^2-\partial_k+XY.
$$
Substituting $X=\df-n_r$ and $Y=\df-n_c$, we get
$$
w\geq \df^2-\partial_k+(\df-n_r)(\df-n_c),
$$
which proves \eqref{eq:firstB}.

It remains to prove the two other lower bounds. Since every non-zero
row and every non-zero column of $c'$ is a codeword in a local RS code
of minimum distance $\delta=\df-r+1$, we clearly have
$
w\geq n_r\delta, n_c\delta.
$
Hence
$$
w\geq
\max\Big\{\df^2-\partial_k+(\df-n_r)(\df-n_c),\; n_r\delta,\; n_c\delta\Big\}
=
B(\partial_k,n_r,n_c).
$$
Since this holds for every pair $(n_r,n_c)$ attainable in $C_k$, taking the
minimum over all $\delta\leq n_r,n_c\leq \df$ gives
$$
d_k\geq \min_{\delta\leq n_r,n_c\leq \df} B(\partial_k,n_r,n_c).
$$
This completes the proof.
\end{proof}

Let us now verify that the bound of Theorem \ref{thm:boundopt} uniformly outperforms the naive RS-type degree lower bound

\begin{proposition}\label{prop:surp}
For all $1\leq k\leq r^2$, the lower bound of Theorem \ref{thm:boundopt} on $d_k$ is not smaller than $n^2-\partial_k$.
\end{proposition}

\begin{proof}
Using the terminology of Theorem \ref{thm:boundopt}, this follows immediately from $(\df-n_r)(\df-n_c)\geq 0$ for all $n_r,n_c\leq \df$.
\end{proof}

Next, we would like to find the exact value of the bound \eqref{eq:boundopt} for $r=k^2-1$ and $r=k^2-2$. For this, it will be useful to first explicitly state the three largest numbers in the set $D$ of degrees.\footnote{These values can also be derived from Proposition \ref{prop:degrees}; however, presenting the calculation directly from the set $D$ is more insightful.}
By Theorem~\ref{thm:max-deg-basis},
$$
D=\Big\{t\df+\ell \;\Big|\; 0\le t\le 2r-2,\ 0\le \ell\le r-1-\Big\lceil \frac t2\Big\rceil\Big\}.
$$
For $t=2r-2$, the only possible value is $\ell=0$, so the largest degree in $D$ is
$
\partial_{r^2}=(2r-2)\df.
$
For $t=2r-3$, again $\ell=0$ is the only possible value, and therefore
$
\partial_{r^2-1}=(2r-3)\df.
$
For $t=2r-4$, we have
$$
0\le \ell\le r-1-\left\lceil \frac{2r-4}{2}\right\rceil = 1,
$$
so the corresponding degrees are
$
(2r-4)\df, (2r-4)\df+1.
$
Hence the third-largest degree in $D$ is
$
\partial_{r^2-2}=(2r-4)\df+1.
$

In the following two corollaries to Theorem \ref{thm:boundopt}, we consider the cases of one and two heavy parities ($k=r^2-1$, $k=r^2-2$), and prove the optimality of the code $C_k$ in these cases.

\begin{corollary}[Distance optimality for one heavy parity]\label{coro:hone}
For $k=r^2-1$, it holds that the minimum distance $d_k$ of $C_k$ equals $\delta(\delta+1)$, where $\delta:=\df-r+1$. Hence $d_{r^2-1}$ coincides with the upper bound of Example \ref{example:hone}, and the code $C_{r^2-1}$ is optimal with respect to the minimum distance.
\end{corollary}

\begin{proof}
Here
$
\partial_k=(2r-3)\df.
$
By Theorem \ref{thm:boundopt} and Example \ref{example:hone}, it is sufficient to prove 
$$
\min_{\delta\le n_r,n_c\le \df} B(\partial_k,n_r,n_c)
=
\delta(\delta+1).
$$

To see that this value is attained, take
$
(n_r,n_c)=(\delta,\delta+1).
$
Then
$
\df-\delta=r-1, \df-(\delta+1)=r-2,
$
and therefore the first term of $B$ equals
$$
\df^2-\partial_k+(\df-\delta)(\df-(\delta+1))
=
\df^2-(2r-3)\df+(r-1)(r-2)
=
\delta(\delta+1).
$$
The second and third terms are
$
n_r\delta=\delta^2, n_c\delta=\delta(\delta+1).
$
Hence
$
B(\partial_k,\delta,\delta+1)
=
\delta(\delta+1),
$
and so
$
\min_{\delta\le n_r,n_c\le \df} B(\partial_k,n_r,n_c)
\le \delta(\delta+1).
$
Conversely, assume that for some $(n_r,n_c)$ one has
$
B(\partial_k,n_r,n_c)<\delta(\delta+1).
$
Since $B$ is the maximum of three terms, this implies in particular
$
n_r\delta<\delta(\delta+1), n_c\delta<\delta(\delta+1).
$
As $n_r,n_c\ge \delta$, it follows that necessarily
$
n_r=n_c=\delta.
$
But then the first term of $B$ becomes
$$
\df^2-\partial_k+(\df-\delta)^2
=
\df^2-(2r-3)\df+(r-1)^2
=
\delta(\delta+1)+(r-1)>\delta(\delta+1).$$ This contradiction proves that
$$
\min_{\delta\le n_r,n_c\le \df} B(\partial_k,n_r,n_c)
=
\delta(\delta+1).
$$
This completes the proof.
\end{proof}

\begin{corollary}[Distance optimality for two heavy parities]\label{coro:htwo}
Suppose that $r\geq 3$. Then for $k=r^2-2$, it holds that the minimum distance $d_k$ of $C_k$ equals $\delta(\delta+2)$, where $\delta:=\df-r+1$, and the code $C_{r^2-2}$ is optimal with respect to the minimum distance.
\end{corollary}

\begin{proof}
Let us first verify that the lower bound of Theorem \ref{thm:boundopt} equals $\delta(\delta+2)$ for $k=r^2-2$. Here
$
\partial_k=(2r-4)\df+1.
$
We claim that
$$
\min_{\delta\le n_r,n_c\le \df} B(\partial_k,n_r,n_c)
=
\delta(\delta+2).
$$
To see that this value is attained, take
$
(n_r,n_c)=(\delta+1,\delta+1).
$
Then
$
\df-(\delta+1)=r-2$,
and so the first term of $B$ equals
$$
\df^2-\partial_k+(\df-(\delta+1))^2
=
\df^2-((2r-4)\df+1)+(r-2)^2
=
\delta(\delta+2).
$$
The second and third terms are both equal to
$
(\delta+1)\delta.
$
Hence
$
B(\partial_k,\delta+1,\delta+1)=\delta(\delta+2)$,
and therefore
$
\min_{\delta\le n_r,n_c\le \df} B(\partial_k,n_r,n_c)
\le \delta(\delta+2).
$

Conversely, assume that for some $(n_r,n_c)$ one has
$
B(\partial_k,n_r,n_c)<\delta(\delta+2).
$
Then
$
n_r\delta<\delta(\delta+2), n_c\delta<\delta(\delta+2),
$
Hence we obtain
$$
\delta \leq n_r,n_c\le \delta+1.
$$
Now write
$
a:=\df-n_r,\qquad b:=\df-n_c.
$
Then
$
a,b\ge \df-(\delta+1)=r-2.
$
Then the first term in $B$ is
$$
\df^2-\partial_k+ab
=
\df^2-((2r-4)\df+1)+ab
=
\delta(\delta+2)+ab-(r-2)^2.
$$
Since $a,b\ge r-2$, $
ab\ge (r-2)^2
$,
and therefore
$
\df^2-\partial_k+ab\ge \delta(\delta+2).
$
Thus the first term of $B$ is already at least $\delta(\delta+2)$, contradicting the assumption
$$
B(\partial_k,n_r,n_c)<\delta(\delta+2).
$$
We conclude that
$$
\min_{\delta\le n_r,n_c\le \df} B(\partial_k,n_r,n_c)
=
\delta(\delta+2).
$$

To prove that the lower bound is equal to the minimum distance and that $C_{r^2-2}$ is optimal with respect to the minimum distance, it is sufficient to prove that the lower bound coincides with an upper bound. Now, in the bound \eqref{eq:gen_bound} for $k=r^2-2$, we may choose
$(a,b)=(1,3)$ (as $r\geq 3$) to obtain $d\leq (\df-r+1)(\df-r+3)=\delta(\delta+2)$. 
\end{proof}

We remark that the above corollary holds also for $r=2$, in which case $2=r^2-2\leq 2r-1=3$, and the assertion follows from Proposition \ref{prop:smallk} below.

So far, we have proved optimality for one and two heavy parities, that is, for very high dimension. Next, we will show that the codes are optimal also for $k\leq 2r-1$. In this case, the 
minimum distance coincides with the upper bound \eqref{eq:1dlrc} for
LRCs with a single recovery set, as shown in the following
proposition. 

\begin{proposition}\label{prop:smallk}
For $k\leq 2r-1$, the minimum distance $d_k$ of $C_k$ satisfies
\begin{equation}\label{eq:1dlrcagain}
d_k= \df^2-k+1-\left\lfloor \frac{k-1}{r} \right\rfloor (\df-r).
\end{equation}
Hence, for such $k$, the minimum distance  coincides with 
the upper bound \eqref{eq:1dlrc}, and the codes are optimal
with respect to the minimum distance.
\end{proposition}

\begin{proof}
In \eqref{eq:d1d2}, the first two intervals in $D$ (corresponding to
$t=0,1$) include 
exactly $2r-1$ degrees, of which, the first interval ($t=0$) is
$\{0,1,\ldots,r-1\}$. So, for $1\leq k\leq r$ we have
$\partial_k=k-1$,  and the RS-type degree
bound coincides with the Singleton bound, which obviously agrees with
\eqref{eq:1dlrcagain}. It therefore remains to consider the range
$r+1\leq k\leq 2r-1$, in the second interval ($t=1$).\footnote{
Note that this set of $k$-values is empty for $r=1$, so that we may
assume $r\geq 2$.
} 
The last degree
in this interval is $\partial_{2r-1}=\df+r-2$, so that the RS-type
degree bound on the minimum distance is $\df^2-\df-r+2$, which agrees
with \eqref{eq:1dlrcagain}. Now, moving downward from $k=2r-1$ to
$k=r+1$ in unit steps, the degrees in the interval decrease in unit
steps, so that the RS-type bound increases in unit steps. The same is
true also for \eqref{eq:1dlrcagain}, as the term $\left\lfloor
\frac{k-1}{r} \right\rfloor$ remains unchanged and equals $1$.
\end{proof}

Although Theorem~\ref{thm:boundopt} provides a stronger lower bound, it
is still useful to understand the behavior of the simpler RS-type degree bound
$$
d_k^{(1)}:=\df^2-\partial_k.
$$
The following proposition gives a closed-form bound at
the breakpoint dimensions $k_t$ of Proposition~\ref{prop:degrees}, and
will be used to compare it analytically with the upper bound of
Proposition~\ref{prop:bound}.


\begin{proposition}\label{prop:lbonlb}
For a positive integer $k\leq r^2$, let $d^{(1)}_k:=\df^2-\partial_k$ be
the RS-type degree lower bound on the 
minimum distance. Then for integer $0\leq t\leq 2r-2$ and for the
dimensions $k_t$ of Proposition \ref{prop:degrees}, it holds that 
\begin{equation}\label{eq:profile}
d^{(1)}_{k_t}\geq  \df^2 -2r\df\Big(1-\sqrt{1-\frac{k_t}{r^2}} \Big)=
\Big(\df-\big(r-\sqrt{r^2-k_t}\big) \Big)^2-\big(r-\sqrt{r^2-k_t}\big)^2.  
\end{equation}
\end{proposition}

\begin{proof}
Since
$$
\left\lfloor \frac{t+1}{2}\right \rfloor \cdot
\left\lceil\frac{t+1}{2} \right\rceil
\leq \frac{(t+1)^2}{4}
$$
(with equality for odd $t$ and strict inequality for even $t$), we
have
$$
k_t\geq \hat{k}_t:=r(t+1)-\frac{(t+1)^2}{4}.
$$
Also, by Proposition \ref{prop:degrees},
$$
d^{(1)}_{k_t}=
\df^2-\left((t+1)\df-\df+r-1-\left\lfloor\frac{t+1}{2}\right\rfloor 
\right)\geq \df^2-(t+1)\df=:\hat{d}_t.
$$
As $t\leq 2r-2$ and $\hat{d}_t$ is monotonically decreasing in $t$, it
holds that for all relevant $t$,
\begin{equation}\label{eq:dt}
\hat{d}_t\geq \hat{d}_{2r-2}=\df^2-(2r-1)\df\geq
\df^2-2r\df. 
\end{equation}

By definition, $t+1=\df-\hat{d}_t/\df$. Substituting into the expression for
$\hat{k}_t$, results in
$$
\hat{k}_t=r\left(\df-\frac{\hat{d}_t}{\df}\right) -
\frac{1}{4}\left(\df-\frac{\hat{d}_t}{\df}\right)^2 
 = r\df-r\frac{\hat{d}_t}{\df}-\frac{\df^2}{4} +
\frac{\hat{d}_t}{2}-\frac{\hat{d}_t^2}{4\df^2}. 
$$
Hence,
$$
\hat{d}_t^2-4\df^2\left(\frac{1}{2}-\frac{r}{\df}
\right)\hat{d}_t+4\df^2\left(\frac{\df^2}{4}-r\df+\hat{k}_t\right) =0.
$$
Solving for $\hat{d}_t$, we obtain
$$
\hat{d}_t=2\df^2\left(\frac{1}{2}-\frac{r}{\df}\right)\pm
2\sqrt{\df^4\left(\frac{1}{2}-\frac{r}{\df}\right)^2-
\df^2\left(\frac{\df^2}{4}-r\df+\hat{k}_t\right)} =\df^2-2r\df\pm
2r\df\sqrt{1-\frac{\hat{k}_t}{r^2}}. 
$$
Since $\hat{dt}\geq \df^2-2r\df$ for all relevant $t$ by
\eqref{eq:dt}, ``$\pm$'' can be replaced by ``$+$'' in the last
equation.  

The function 
$$
u\colon x\mapsto \df^2
-2r\df\Big(1-\sqrt{1-\frac{x}{r^2}} \Big)
$$
is decreasing for $x\in (0,r^2)$, and we have just proved that
$\hat{d}_t=u(\hat{k}_t)$. Therefore,
$$
d^{(1)}_{k_t}\geq \hat{d}_t = u(\hat{k}_t) \geq u(k_t),
$$
(using $k_t\geq \hat{k}_t$), and we are done.
\end{proof}

It is now useful to compare \eqref{eq:profile} to the upper bound
\eqref{eq:bound}. For $t$ as in Proposition~\ref{prop:lbonlb}, we have
\begin{equation}\label{eq:ccc}
\Big(\sqrt{r^2-k_t}+\df-r
\Big)^2-\big(r-\sqrt{r^2-k_t}\big)^2 \leq d_{k_t}\leq \Big(\Big\lceil
\sqrt{r^2-k_t+1}\Big\rceil+\df-r \Big)^2.
\end{equation}

Thus, although the RS-type degree lower bound is in general not tight
for dimensions close to $r^2$, \eqref{eq:ccc} shows that for smaller
dimensions it gets close to the upper bound. To make this precise, let
$
A:=\Big(\sqrt{r^2-k_t}+\df-r\Big)^2,
B:=\big(r-\sqrt{r^2-k_t}\big)^2,
$
so that the lower bound in \eqref{eq:ccc} is $A-B$, and let
$
C:=\Big(\Big\lceil\sqrt{r^2-k_t+1}\Big\rceil+\df-r \Big)^2
$
be the upper bound.

It holds that, 
$$
\frac{A}{B}=\left(\frac{\df/r}{1-\sqrt{1-\frac{k_t}{r^2}}}-1\right)^2.
$$
This expression is increasing in $\df/r$, and decreasing in
$k_t/r^2$. Hence, if, for example, $\df/r\geq 2$ (local RS codes have
rate at most $1/2$), and $k_t/r^2\leq 0.33$, then we get 
$$
\frac{A}{B}\geq \left(\frac{2}{1-\sqrt{1-0.33}}-1\right)^2\approx
100.
$$
In particular,
$
A-B\geq 0.99A.
$

Also, 
we have
\begin{equation}\label{eq:CoverA}
\frac{C}{A}\leq
\left(\frac{\sqrt{A}+1}{\sqrt{A}}\right)^2 =
\left(1+\frac{1}{\sqrt{A}}\right)^2\leq
\left(1+\frac{1}{0.9\df}\right)^2,
\end{equation}
where in the last inequality we have used the assumptions $\df/r\geq 2$ and $k_t/r^2\leq 0.33$.
For example, when $\df\geq 32$, the right-hand side of \eqref{eq:CoverA} is at most about
$1.1$.

Combining the above estimates gives the following corollary.

\begin{corollary}\label{cor:ninety-percent}
Assume that
$$
\frac{\df}{r}\geq 2,\qquad \frac{k_t}{r^2}\leq 0.33,\qquad
\df\geq 32.
$$
Then
$$
d_{k_t}\geq 0.9\cdot
\Big(\Big\lceil\sqrt{r^2-k_t+1}\Big\rceil+\df-r \Big)^2.
$$
In other words, for such parameters, the lower bound of
Proposition~\ref{prop:lbonlb} is at least $90\%$ of the upper bound of
Proposition~\ref{prop:bound}.
\end{corollary}

\begin{proof}
By Proposition~\ref{prop:lbonlb} and \eqref{eq:ccc},
$$
d_{k_t}\geq A-B.
$$
From the discussion above,
$
A-B\geq 0.99A
$
and
$
C\leq 1.1A.
$
Hence
$$
d_{k_t}\geq 0.99A \geq \frac{0.99}{1.1}C = 0.9C.
$$
\end{proof}

We stress that Corollary~\ref{cor:ninety-percent} is only meant to
illustrate analytically that the construction is provably close to the
upper bound over a broad parameter range. In concrete instances, the
stronger lower bound of Theorem~\ref{thm:boundopt} is typically better, as demonstrated in Section~\ref{sec:eg}.

\section{Instantiation}\label{sec:inst}
Let us now consider a concrete choice for $f$ such that the splitting
field of $f$ and $g=x-f$ is small. Let $q$ be any prime power. Take
some $c\in \ef_{q^2}\smallsetminus \efq$, and note that $c^{q-1}\neq
0,1$. Now define $f_1(x)=x^q-x$, and
$$
g_1(x):=c^q\Big(\Big(\frac{x}{c} \Big)^q-\frac{x}{c}\Big)=x^q-c^{q-1}x.
$$
Then $f_1-g_1=(c^{q-1}-1)x$, so that letting 
$$
f:=\frac{1}{c^{q-1}-1}f_1=\frac{1}{c^{q-1}-1}(x^q-x),
$$
and
$$
g:=-\frac{1}{c^{q-1}-1}g_1 =
\frac{c^{q-1}}{c^{q-1}-1}x-\frac{1}{c^{q-1}-1}x^q, 
$$
we have $f+g=x$. Also, $Z_f=\efq$, while $Z_g=Z_{g_1}=c\efq\subset\ef_{q^2}$. Hence, we may take
$F=\ef_{q^2}$. Since $Z_f$ and $Z_g$ are distinct $1$-dimensional
$\efq$-subspaces of $\ef_{q^2}$, we have $Z_f\cap Z_g=\{0\}$, and therefore
$$
Z_f\oplus Z_g=\ef_{q^2}.
$$
In this instantiation, $\df=\deg(f)=q$, so the codes $C_k$ have length
$
\df^2=q^2=|F|.
$

\section{Examples}\label{sec:eg}
In this section we include some concrete examples for comparing the
upper bounds with the lower bound for the codes $C_k$. In all examples, we take $q$ to be a power of $2$.

\begin{example}
{\rm
In this example, we examine codes of very low rate. Take $\df=32$
and $r=8$. Hence, $C_{r^2}$ is the 
product of two $[32,8]$ RS codes of rate $1/4$, and has rate
$1/16$. Since here $\df=q=32=2^5$, Section~\ref{sec:inst} allows us to work over
$\ef_{q^2}=\ef_{2^{10}}$. Figure \ref{fig:eg1} compares the lower bound \eqref{eq:boundopt}, the upper bound \eqref{eq:gen_bound}, and the upper
bound \eqref{eq:boundv2} of Appendix \ref{app:boundv2}. For 
\eqref{eq:gen_bound}, the minimum over all possible $0\leq a,b\leq r$
was taken.\footnote{We 
comment that it is certainly not required to exhaustively scan over
$r^2$ elements. First, the problem can be easily converted into the
following one dimensional minimization:
$$
\min_{\Big\lceil\frac{r^2-k+1}{r} \Big\rceil\leq a\leq
r}(a+\df-r)\left(\Big\lceil\frac{r^2-k+1}{a}\Big\rceil+\df-r \right). 
$$
Also, the minimization can probably be
further simplified, but since we currently do not see an informative
closed-form expression for the minimum, we will not get into further
details.} 
We see that the lower bound \eqref{eq:boundopt} is not far from the lowest of the upper
bounds throughout the $k$ range. In addition, the lower bound almost coincides with the lowest of  the upper bounds
for, say, $k\leq 30$. We also see that the upper bound
\eqref{eq:gen_bound} is significantly tighter than the bound
\eqref{eq:boundv2} of the appendix for the higher values of $k$, and
slightly less tight for low values of $k$.

\begin{figure}
\begin{subfigure}[h]{0.5\linewidth}
\includegraphics[width=\linewidth]{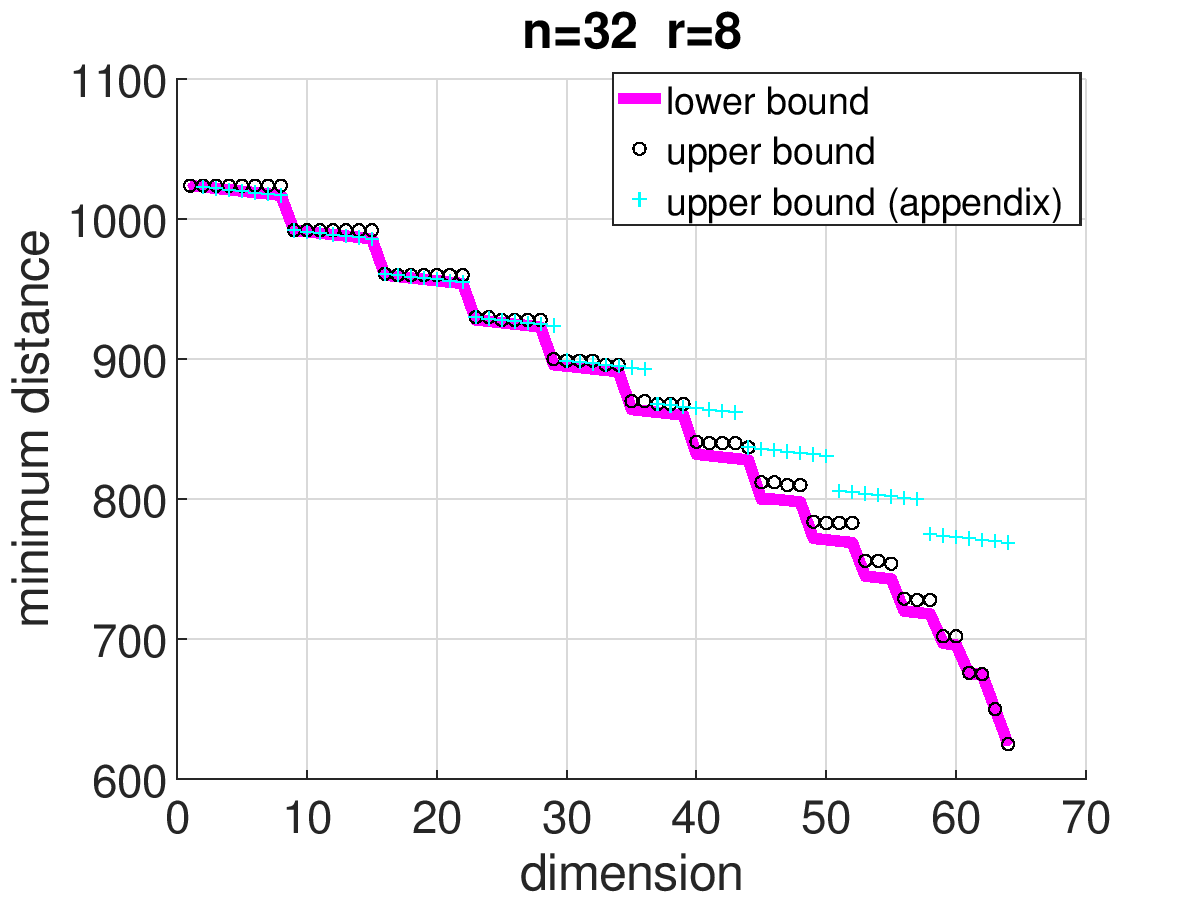}
\caption{Entire $k$ range.}
\end{subfigure}
\hfill
\begin{subfigure}[h]{0.5\linewidth}
\includegraphics[width=\linewidth]{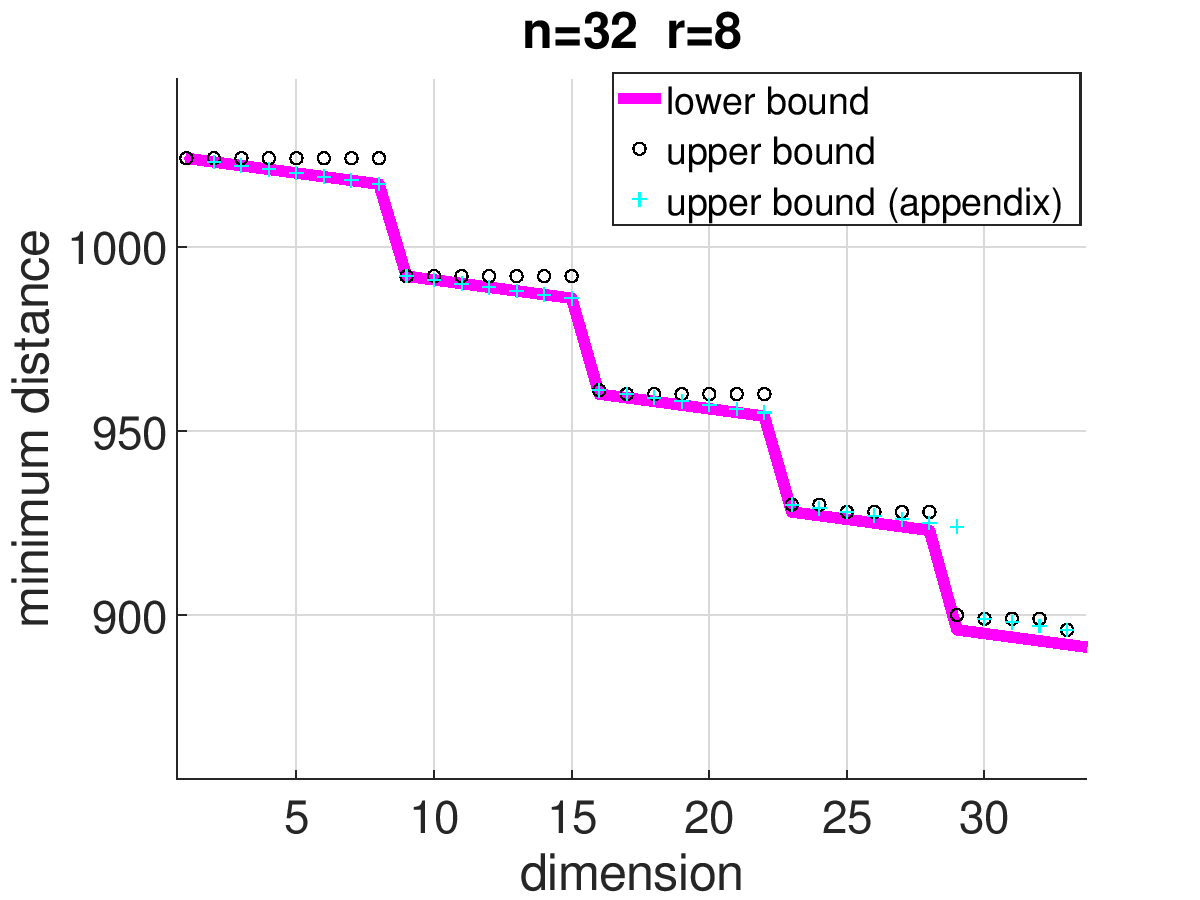}
\caption{Zoom on low $k$ range.}
\end{subfigure}%
\caption{A comparison of the lower bound \eqref{eq:boundopt}, the upper
bound \eqref{eq:gen_bound} (minimum over $a,b$), and the upper bound 
\eqref{eq:boundv2} of Appendix \ref{app:boundv2} for
$(\df,r)=(32,8)$.}\label{fig:eg1} 
\end{figure}
}
\end{example}

\begin{example}
{\rm
Here, we consider two cases where the local codes have
rate $1/2$, so that $C_{r^2}$ has rate $1/4$. In the first case,
$(\df,r)=(32,16)$, while in the second, $(\df,r)=(128,64)$, as in a
suggested coding scheme for data availability sampling in blockchains \cite{AlBassam2018}.\footnote{In this scheme, both tensored RS codes have rate $1/2$. The specific case of a local length of $128$ is one of the examples considered in \cite{AlBassam2018}.} The
underlying finite field can be $\ef_{2^{10}}$ and $\ef_{2^{14}}$,
respectively. Figure \ref{fig:eg2}  
compares the lower bound  \eqref{eq:boundopt}, the upper bound
\eqref{eq:gen_bound}, and the upper bound 
\eqref{eq:boundv2} of Appendix \ref{app:boundv2} for this case.  
For low values of the dimension $k$, the lower bound is very close to the
upper bound. In addition, the exact values at $k=r^2-1$ and $k=r^2-2$
coincide with the upper bound. Also,
even for lower $k$ values, the lower bound is 
not too far from the upper bound. We are not aware of previous explicit constructions in the specific setting considered here, namely explicit large-distance subcodes of RS product codes over a field whose size is linear in the block length. However,
it is possible to compare our construction with the
obvious option of product subcodes of the form $C_1'\otimes C_2'$, where $C_i'$ is a subcode of the $i$-th component RS code. Considering the case of
$(\df,r)=(128,64)$, the highest dimension achievable in
this case by a proper product subcode is $63\cdot 64=4032$, with
distance $(128-63+1)\cdot(128-64+1)=4290$ (compared to $4225$ of the
full product code). In comparison, for a dimension of $4032$, the lower
bound \eqref{eq:boundopt} equals $4940$. We comment that for product
subcodes of lower dimensions, the gain from the current construction
gets increasingly larger. 

\begin{figure}
\begin{subfigure}[h]{0.5\linewidth}
\includegraphics[width=\linewidth]{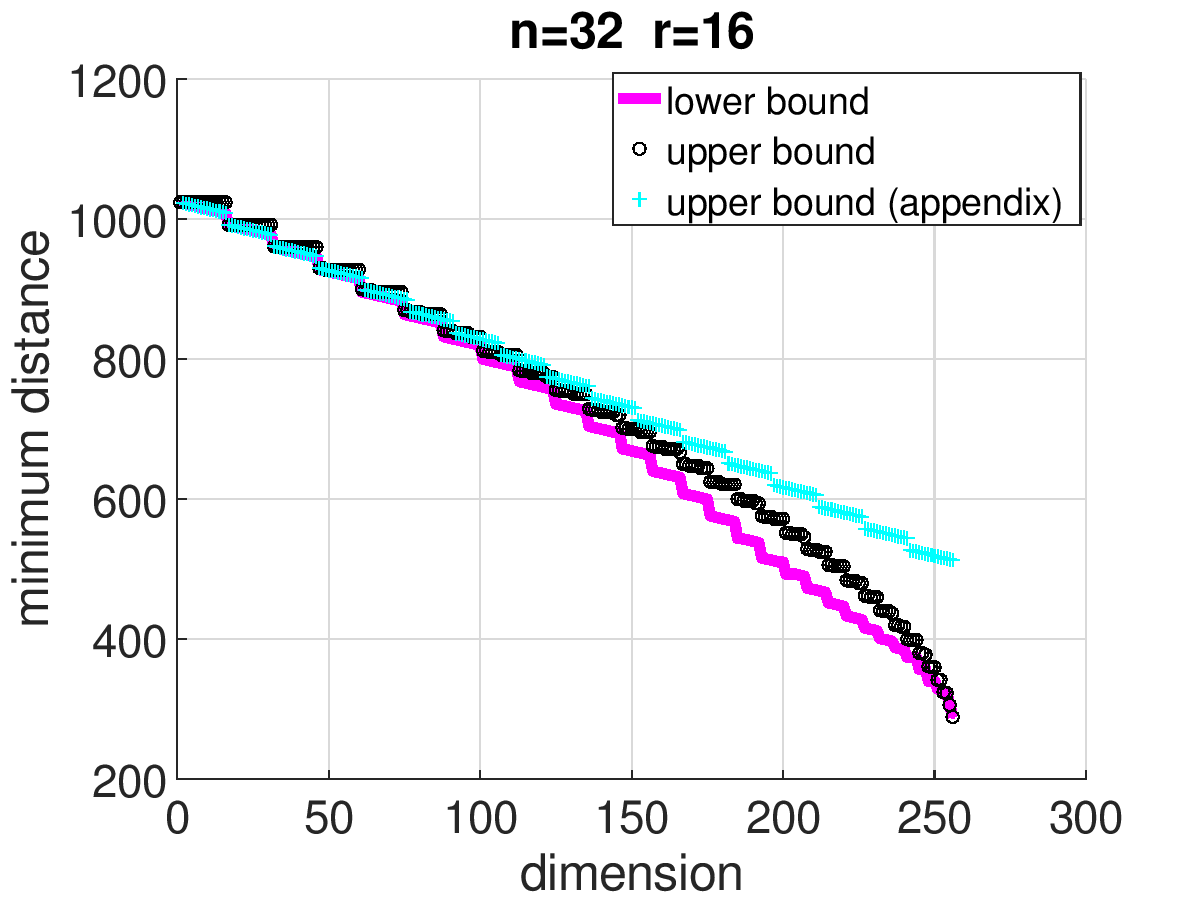}
\caption{Local length $32$, local dimension $16$.}
\end{subfigure}
\hfill
\begin{subfigure}[h]{0.5\linewidth}
\includegraphics[width=\linewidth]{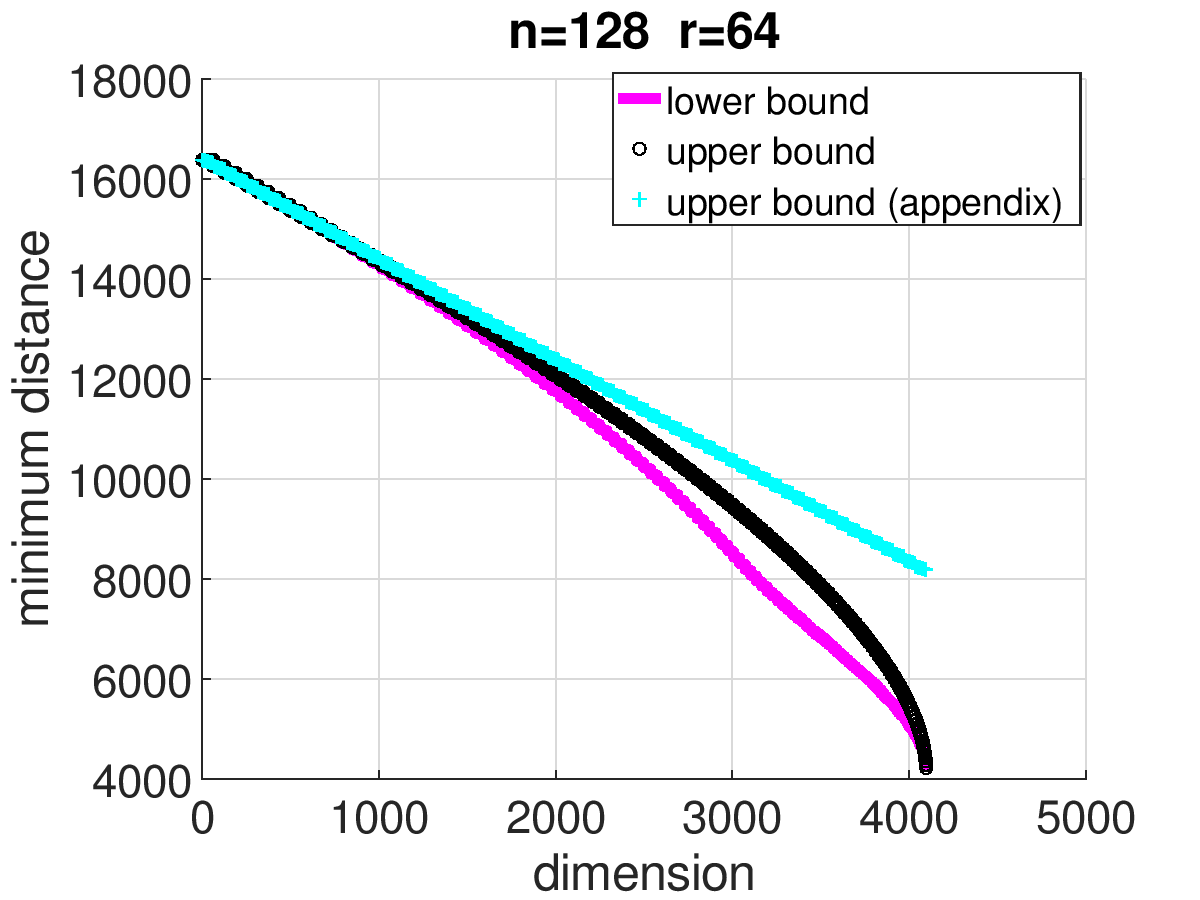}
\caption{Local length $128$, local dimension $64$.}
\end{subfigure}%
\caption{A comparison of the lower bound \eqref{eq:boundopt}, the upper
bound \eqref{eq:gen_bound}  (minimum
over $a,b$), and the upper bound \eqref{eq:boundv2} of Appendix
\ref{app:boundv2} for (a) $(\df,r)=(32,16)$, (b) $(\df,r)=(128,64)$.}
\label{fig:eg2} 
\end{figure}
}
\end{example}

\begin{example}
{\rm
Let us now consider the case where $(\df,r)=(32,25)$, over
$\ef_{2^{10}}$. In this case, 
$C_{r^2}$ has rate slightly larger than $0.6$. 
Figure \ref{fig:eg3}
compares the lower bound \eqref{eq:boundopt}, the upper bound
\eqref{eq:gen_bound}, and the upper bound 
\eqref{eq:boundv2} of Appendix \ref{app:boundv2} for this case. 
Even for the high rate of this example, we see that the lower bound is
close to the upper bound for low values of the dimension $k$. However,
for values of $k$ close to $r^2$, the lower bound is quite far from the upper
bound, except for $k=r^2-2,r^2-1$, where the construction is
optimal. It remains an open question whether the gap is because the
lower or upper bounds are not tight in this case, or because the codes
are indeed far from the highest achievable distance. In the latter
case, it would be interesting to see if better results can be obtained
when using polynomials modulo the annihilator $a(x)$ of $Z_f\oplus
Z_g$, as in Remark \ref{rem:modulo}.

\begin{figure}
\centering
\includegraphics[width=0.5\linewidth]{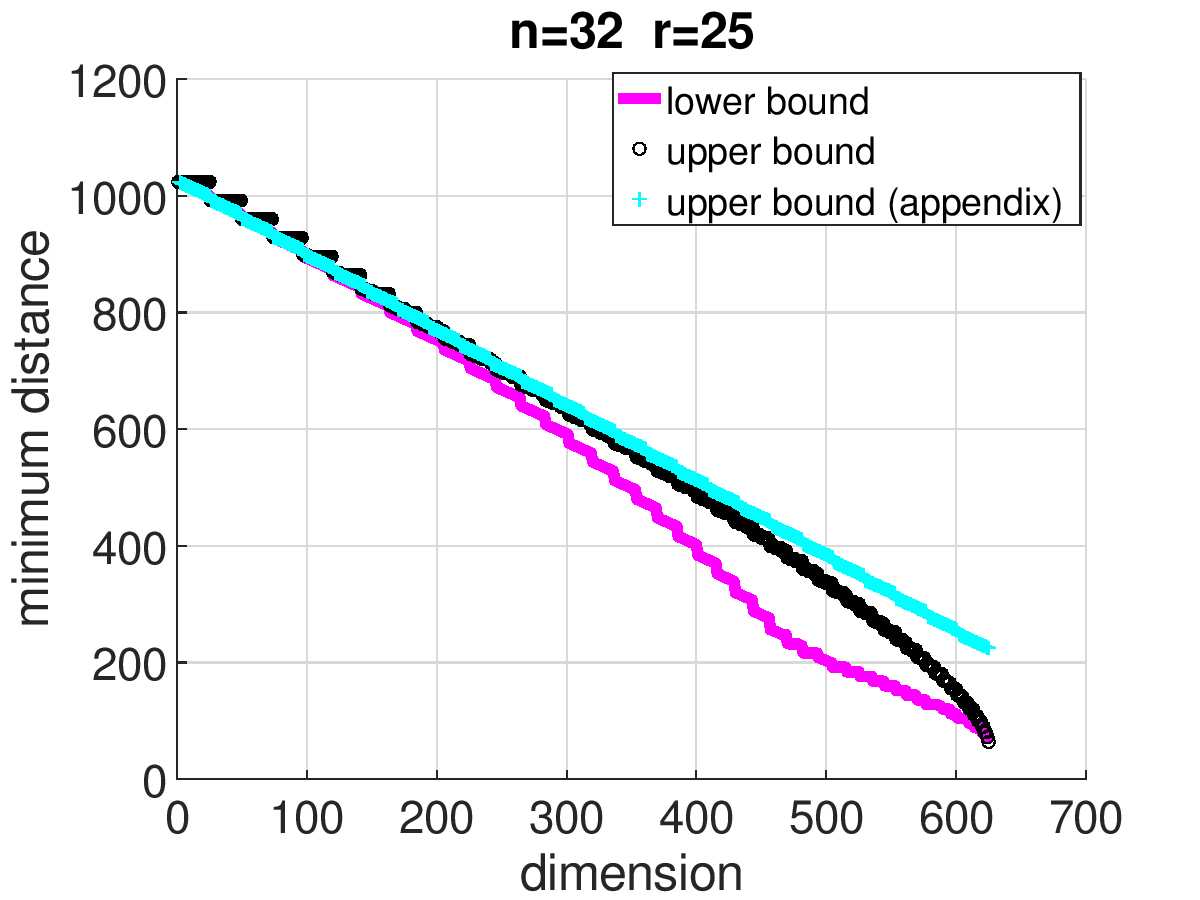}
\caption{A comparison of the lower bound \eqref{eq:boundopt}, the upper bound \eqref{eq:gen_bound} (minimum
over $a,b$), and the upper bound
\eqref{eq:boundv2} of Appendix \ref{app:boundv2} for
$(\df,r)=(32,25)$.} \label{fig:eg3}
\end{figure}
}
\end{example}

\section{Conclusions and open questions}
\label{sec:conc}
We presented an explicit construction of subcodes of the product of
two $[\df,r]$ Reed--Solomon codes (where $\df$ is a prime power) that attains the maximum possible
minimum distance for dimensions $r^2-1$, $r^2-2$, and all dimensions at most 
$ 2r-1$. Moreover, when $r/\df \leq 1/2$, the minimum distance
of our construction is provably close to optimal over a broad range of
dimensions; for example, it attains at least $90$\% of the maximum
possible distance at interval endpoints when $k/r^2 \leq 0.33$. Our construction requires a field whose size is bounded linearly by the overall product code length. Complementing this explicit construction, we established a new upper bound on the minimum distance that explicitly captures the rigid intersecting grid topology of product codes for arbitrary local distances.

Several interesting questions remain open for future research. First,
the tightness of the upper bound \eqref{eq:gen_bound} calls for further
investigation. The results of Section \ref{sec:5} show that this
bound is attained for dimensions $k = r^2, r^2-1, r^2-2$, and we
believe that it remains nearly tight for values of $k$ close to
$r^2$. 
More generally, we speculate that the minimum of the upper bound \eqref{eq:gen_bound} and the bound \eqref{eq:boundv2} from Appendix \ref{app:boundv2} is close to the maximum attainable distance for all dimensions $k$ at which this minimum strictly decreases from its value at $k-1$.

A second question concerns the tightness of the lower bound in
Theorem \ref{thm:boundopt} for the codes $C_k$. Are there values of $k$ for
which the true minimum distance of $C_k$ is strictly larger than this
bound? 
It would also be of interest to derive general lower bounds,
possibly via probabilistic methods. Notably, the explicit construction
presented here requires a field size comparable to the code length. It is natural to ask whether permitting larger field sizes can lead to improved distances for certain dimensions, in analogy with the construction in~\cite{WZ14} for local distance~$2$, where a field size exponential in the code length is required.

Finally, the problem of efficient decoding up to half the minimum
distance remains open. The product code $C_{r^2}$ can be decoded up to
half its minimum distance using generalized minimum distance (GMD)
decoding \cite{forneyGMD}; see \cite[Sec.~12.2]{roth} and
\cite[pp.~178--179]{roth}. It is an interesting question whether
similarly efficient decoding algorithms exist for $C_k$ when $k <
r^2$. 

\appendix

\section{Proof of Theorem \ref{thm:max-deg-basis}}\label{app:proof}
We will
first need the following lemma. 

\begin{lemma}\label{lemma:sum}
For $t\in\{0,1,\ldots,2r-2\}$, let $I_t:=\{t\df+\ell |0\leq \ell\leq
r-1-\lceil t/2\rceil\}$. Then 
\begin{equation}\label{eq:sum}
\Big|\bigcup_{j=0}^t I_j\Big|= \sum_{j=0}^t |I_j| =
(t+1)r-\left\lfloor \frac{t+1}{2}\right \rfloor 
\cdot \left\lceil\frac{t+1}{2} \right\rceil.
\end{equation}
\end{lemma}

\begin{proof}
As $r\leq \df$, distinct sets $I_j$ are disjoint, and the first 
equality in \eqref{eq:sum} follows. Next, noting that $\lceil
t/2\rceil\leq r-1$ so that no $I_j$ is empty,
$$
\sum_{j=0}^t |I_j| 
 = \sum_{j=0}^t \Big(r- \left\lceil \frac{j}{2} \right\rceil\Big)
 = (t+1)r-\sum_{j=1}^t  \left\lceil \frac{j}{2} \right\rceil.
$$
Now, if $t$ is even, then
$$
\sum_{j=1}^t \left\lceil \frac{j}{2} \right\rceil =
1+1+2+2+\cdots+\frac{t}{2}+\frac{t}{2}=\frac{t}{2}\Big(\frac{t}{2}+1\Big),
$$
while if $t$ is odd, we can use the above expression to sum up to
$t-1$ (even when $t-1=0$), and add $(t+1)/2$ to obtain
$$
\frac{t-1}{2}\Big(\frac{t-1}{2}+1
\Big)+\frac{t+1}{2}=\Big(\frac{t+1}{2}\Big)^2. 
$$
This completes the proof.
\end{proof}

We can now turn to the proof of the theorem.
\begin{proof}[Proof of Theorem \ref{thm:max-deg-basis}]
By Lemma \ref{lemma:sum}, the
number of elements on the right-hand side of \eqref{eq:d1d2} is
$r^2=|B_{r^2}|=|D|$, and it is therefore sufficient to prove
that the right-hand side of \eqref{eq:d1d2} is contained in $D$. 
Let 
$$
D_1:=\big\{t \df + s-t \big|\ 0 \leq t \leq s \leq r-1 \big\}, 
$$
and
$$
D_2:=\big\{t \df + s-t\big| r\leq s\leq 2(r-1),\ 2(s - r+1)
\leq t \leq s \big\}. 
$$
We will first prove that 
\begin{equation}\label{eq:supset}
D\supseteq D_1\cup D_2,
\end{equation} 
and then verify
that the right-hand side of \eqref{eq:d1d2} equals $D_1\cup D_2$.

Define $P_{ij}:=f(x)^ig(x)^j=f(x)^i(x-f(x))^j$ for $i,j\in\{0,\ldots,
r-1\}$. All polynomials $P_{ij}$ with the same sum of indices $s:=i+j$
have the same degree. We will show that $D_1\cup D_2$ is the set of
degrees obtained when putting all ``slices'' $X_s:=\{P_{ij}|0\leq i,j\leq
r-1, i+j=s\}\subseteq B_{r^2}$, $0\leq s\leq 2(r-1)$, in REF. 

Since
$$
f^i(x-f)^j=\sum_{\ell=0}^j
\binom{j}{\ell}(-1)^{j-\ell}x^{\ell}f^{i+j-\ell}
=\sum_{\ell=0}^j\binom{j}{\ell}(-1)^{j-\ell}x^{\ell}f^{s-\ell}, 
$$
it is clear that for all $s$, 
\begin{equation}\label{eq:xs}
X_s\subseteq\linspan_F\{f^tx^{s-t}|0\leq t\leq s\}.
\end{equation}
Write $Y_s:=\{f^tx^{s-t}|0\leq t\leq s\}$ for short. We note that the
polynomials in $Y_s$ have distinct degrees, and are therefore linearly
independent.   

Let us now distinguish between two cases. If $s\leq r-1$, then the set
of pairs $(i,j)$ with $i,j\leq r-1$, $i+j=s$, is
$\{(0,s),(1,s-1),\ldots, (s,0)\}$, and has cardinality
$s+1=|Y_s|$. Hence, $|X_s|=|Y_s|$, and since
$X_s$ is linearly independent, it follows from \eqref{eq:xs} that
$\linspan_F(X_s)=\linspan_F(Y_s)$. Running over all $0\leq s\leq r-1$,
this accounts exactly for the degrees specified in $D_1$.

In the second case, $r\leq s\leq 2(r-1)$. Here, the set of pairs
$(i,j)$ with $i,j\leq r-1$, $i+j=s$, is 
$\{(s-(r-1),r-1),(s-(r-1)+1,r-2),\ldots,(r-1,s-(r-1)\}$, and has
cardinality $2r-1-s\leq r-1<s+1=|Y_s|$. Hence, we have a strict inclusion
$\linspan_F(X_s)\subset\linspan_F(Y_s)$. 

Consider the representation of a
non-zero polynomial in $\linspan_F(X_s)$ as a linear combination of
elements of $Y_s$. We
claim that this representation must include an element $f^ix^{s-i}$ for
some $i\geq s-|X_s|+1=2(s-r+1)$ with a non-zero coefficient. Indeed, assume by
contradiction that there exist scalars $\alpha_t\in F$, $s-(r-1)\leq
t\leq r-1$, not all zero, and $\beta_i\in F$, $0\leq i\leq 2s-2r+1$,
such that 
\begin{equation}
\label{eq:stam}
\sum_{t=s-(r-1)}^{r-1}\alpha_tP_{t,s-t}=\sum_{i=0}^{2s-2r+1}\beta_if^ix^{s-i}.
\end{equation}
Since $X_s$ is linearly independent and the $\alpha_t$ are not all zero, the left-hand side is a non-zero polynomial. Consequently, the right-hand side is also non-zero and has degree at most $(2s-2r+1)\df+2r-s-1$. However, the
left-hand side is divisible by $\big(f(x-f)\big)^{s-(r-1)}$, and
therefore has degree at least 
$$
2(s-r+1)\df=(2s-2r+1)\df+\df>(2s-2r+1)\df+r-1\geq (2s-2r+1)\df+2r-s-1,
$$
(since $s\geq r$), and we arrive at a contradiction. 

Let $\bs{x}_s$ be a column vector 
 whose entries are the polynomials
in $X_s$, and let $\bs{y}_s:=(x^s,fx^{s-1},\ldots,f^s)^T$ be a column
vector whose entries are the polynomials of $Y_s$ in increasing degree
order. Let $M\in
F^{|X_s|\times |Y_s|}=F^{(2r-1-s)\times (s+1)}$ be the matrix such
that $\bs{x}_s=M\bs{y}_s$, and write $M=(M_1,M_2)$, where $M_2\in
F^{|X_s|\times |X_s|}$. We have just proved that for all non-zero column vector $\bs{u}\in F^{|X_s|}$, in the representation $\bs{u}^T\bs{x}_s=(\bs{u}^T M_1,\bs{u}^T M_2)\bs{y}_s$, the last $|X_s|$ coordinates form a non-zero vector, that is, $\bs{u}^T M_2\neq 0$. It follows that $M_2$ is of full
rank. Considering $M_2^{-1}\bs{x}_s=(M_2^{-1}M_1,I)\bs{y}_s$, it follows that the set of
degrees of $X_s$ in REF is the set of degrees of
$\{f^tx^{s-t}|2(s-r+1)\leq t\leq s\}$. The union of these  sets of
degrees over $r\leq s\leq 2(r-1)$ is exactly $D_2$. This 
concludes the proof of \eqref{eq:supset}. 

Now, the conditions $0\leq t\leq s\leq r-1$ in the definition of $D_1$ are easily seen to be
equivalent to $0\leq t\leq r-1$ and $0\leq s-t\leq r-1-t$, so that 
$$
D_1=\{t\df+\ell|0\leq t\leq r-1, 0\leq \ell\leq r-1-t\}.
$$ 
For a similar description of $D_2$, we note first that in the
definition $D_2$, $t$ can take any value in $\{2,3,\ldots, 2(r-1)\}$.\footnote{The values $2,\ldots,r$ are attained for $s=r$, while the values $r+1,\ldots,2(r-1)$ are attained, e.g., as the largest $t$ for $s=r+1,\ldots,2(r-1)$.} Now
$s\geq r$ and $s\geq t$ is equivalent to $s-t\geq\max\{0,r-t\}$. Also,
the condition $2(s-r+1)\leq t$ is equivalent to $s-t\leq r-1-t/2$. We also note that since 
$r-1+t/2\leq 2(r-1)$ for $t\leq 2(r-1)$, the remaining 
condition $s\leq 2(r-1)$ is automatically satisfied. Hence,
$$
D_2=\Big\{t\df+\ell\Big|2\leq t\leq 2(r-1),\max\{0,r-t\}\leq \ell\leq
r-1-\left\lceil\frac{t}{2}\right\rceil\Big\}.
$$

It is now straightforward to verify that $D_1\cup D_2$ is equal to the
right-hand side of \eqref{eq:d1d2}:
For $t\geq r$, only $D_2$ contributes intervals to the union $D_1\cup
D_2$, and the lower limit of $\ell$ in $D_2$ equals $0$ for such $t$,
as in \eqref{eq:d1d2}. For $2\leq t\leq r-1$, both $D_1$ and 
$D_2$ contribute to the $t$-th interval, where the contribution of
$D_1$ is $\{t\df, \ldots, t\df+r-1-t\}$, while the contribution
of $D_2$ is $\{t\df+r-t,\ldots,t\df+r-1-\left\lceil
\frac{t}{2}\right\rceil\}$. The remaining cases of $t=0,1$, with
contribution only from $D_1$, are also easily verified. 
\end{proof}

\section{An additional upper bound}\label{app:boundv2}

\begin{proposition}\label{prop:boundv2}
Let $C_1,C_2$ be two linear codes of length
$n$ and minimum distance at least $\delta=n-r+1$ (for integers
$2\leq r\leq n-1$), and let 
$C\subseteq C_1\otimes C_2$ be a linear subcode. Write
$k:=\dim(C)\leq r^2$ and $d:=\min.\mathrm{dist}(C)$. Then for $k\geq
2$, it holds that  
\begin{equation}\label{eq:boundv2}
d\leq n^2-k+1-\left\lfloor \frac{k-2}{r-1} \right\rfloor (\delta-1).
\end{equation}
\end{proposition}

\begin{proof}

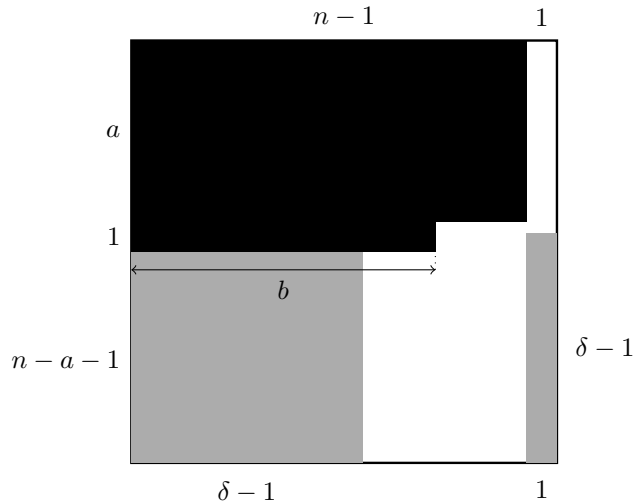
\begin{figure}[h!tbp]
\centering
\begin{tikzpicture}[ 
    every node/.style={font=\footnotesize\sffamily}, 
    thin/.style={line width=0.35pt},
  ]

  \draw[line width=0.9pt] (0,0) rectangle (7*\bbb,7*\bbb);

  \fill[gray!65] (0,0) rectangle (3.8*\bbb,3.5*\bbb);       
  \fill[gray!65] (6.5*\bbb,0) rectangle (7*\bbb,3.8*\bbb);  

  \fill[black] (0,4*\bbb) rectangle (6.5*\bbb,7*\bbb);         
  \fill[black] (0,3.5*\bbb) rectangle (5*\bbb,4*\bbb);   

  \draw[<->] (0,3.2*\bbb) -- (5*\bbb,3.2*\bbb);
  \node[below=0mm of {(2.5*\bbb,3.2*\bbb)}]   {$b$};
  \draw[dotted] (5*\bbb,3.2*\bbb) -- (5*\bbb,3.5*\bbb);

  \node[left=8mm of {(\bbb,5.5*\bbb)}]   {$a$};
  \node[left=8mm of {(\bbb,3.75*\bbb)}]   {$1$};
  \node[left=8mm of {(\bbb,1.7*\bbb)}] {$n-a-1$};

  \node[right=1mm of {(7*\bbb,1.9*\bbb)}]   {$\delta-1$};

  \node[below=1mm of {(1.9*\bbb,0)}]     {$\delta-1$};      
2aaa column
  \node[below=1mm of {(6.75*\bbb,0)}]     {$1$};      

  \node[above=0.5mm of {(6.75*\bbb,7*\bbb)}]     {$1$};      
  \node[above=0.5mm of {(3.5*\bbb,7*\bbb)}]     {$n-1$};      

\end{tikzpicture}
\caption{The erasure pattern for the proof of Proposition \ref{prop:boundv2}.}
\label{fig:patternv2}
\end{figure}

First note that if $2\leq k<r+1$, then \eqref{eq:boundv2} reduces to the
Singleton bound, so that we may assume without loss of generality that
$k\geq r+1$.  

The proof follows the lines of the proof of Proposition
\ref{prop:bound}, but with a different erasure pattern, as described
in Figure \ref{fig:patternv2}. In the figure, $a,b\leq n-1$ are
non-negative integers. 
The black part consists of an $a\times (n-1)$ rectangle together with
an additional black strip of length $b$, so its total size is
$a(n-1)+b$. The gray part consists of a rectangle of size
$(n-a-1)\times (\delta-1)$ together with a column of height
$\delta-1$.

As in the proof of Proposition
\ref{prop:bound}, if an erasure pattern on the black part can be
decoded, then also an erasure pattern on both the black and gray parts
can be decoded (first the right column of height $\delta-1$ is decoded
by local decoding, then the left gray rectangle, row-by-row, by local
decoding, and finally the black part). 

Hence, for non-negative integers $a,b\leq n-1$, if 
$$
\underbrace{a(n-1)+b}_{\mathrm{black}}+\underbrace{(n-a-1)(\delta-1) +
\delta-1}_{\mathrm{gray}}\geq n^2-k+1,
$$
then $d\leq a(n-1)+b$. In other words, for
non-negative integers $a,b\leq n-1$, it holds that 
\begin{equation}\label{eq:impliesv2}
a(r-1)+b\geq nr-k+1\quad \implies\quad d\leq a(n-1)+b.
\end{equation}

To minimize the resulting upper bound $a(n-1)+b$ under the constraints $0\leq a,b\leq n-1$,
$a(r-1)+b\geq nr-k+1$, it is natural to take $a$ as small as possible
and then choose $b$ so that equality holds.
Consider the choice of 
$$
a=n-\left\lfloor\frac{k-2}{r-1} \right\rfloor
$$
and $b=nr-k+1-a(r-1)$.
Note first that since $k-2\geq r-1$ by assumption, it holds that
$a\leq n-1$. Also, as $k\leq r^2$ and $r\leq n$, we have $a\geq 0$. 
In addition, 
\begin{align}
b &= nr-k+1-\Big(n-\left\lfloor\frac{k-2}{r-1} \right\rfloor
\Big)(r-1)\nonumber\\ 
&= n-k+1 +(r-1)\left\lfloor\frac{k-2}{r-1} \right\rfloor\label{eq:bbb}\\
&= n-1-(k-2)+(r-1)\left\lfloor\frac{k-2}{r-1} \right\rfloor\nonumber\\
&= n-1 - (k-2)\bmod(r-1),\nonumber
\end{align}
and therefore $0\leq b\leq n-1$. By the definition of $b$, the
condition in \eqref{eq:impliesv2} is satisfied, and therefore
\begin{align*}
d&\leq \Big(n-\left\lfloor\frac{k-2}{r-1} \right\rfloor\Big)(n-1) +
n-k+1 +(r-1)\left\lfloor\frac{k-2}{r-1}\right\rfloor\\
&=n^2-k+1-\left\lfloor \frac{k-2}{r-1} \right\rfloor (\delta-1),
\end{align*}
where we have used the description of $b$ from \eqref{eq:bbb}. 
\end{proof}

\begin{remark}
{\rm
Clearly, the Singleton-like bound \eqref{eq:1dlrc} for a single
recovering set holds also for the product code of the proposition. The
only difference  is that in \eqref{eq:boundv2}, $\lfloor
(k-2)/(r-1)\rfloor$ appears instead of $\lfloor (k-1)/r\rfloor$. 
}
\end{remark}

\bibliographystyle{alpha}
	\bibliography{bibliography_a}

@article{Elias54,
  author       = {Peter Elias},
  title        = {Error-free Coding},
  journal      = {Trans. {IRE} Prof. Group Inf. Theory},
  volume       = {4},
  pages        = {29--37},
  year         = {1954},
  url          = {https://doi.org/10.1109/TIT.1954.1057464},
  doi          = {10.1109/TIT.1954.1057464},
  bibsource    = {dblp computer science bibliography, https://dblp.org}
}

@article{kathy,
  author       = {Kathryn Haymaker and
                  Beth Malmskog and
                  Gretchen L. Matthews},
  title        = {Locally recoverable codes with availability \emph{t}{\(\geq\)}2 from
                  fiber products of curves},
  journal      = {Adv. Math. Commun.},
  volume       = {12},
  number       = {2},
  pages        = {317--336},
  year         = {2018},
  url          = {https://doi.org/10.3934/amc.2018020},
  doi          = {10.3934/AMC.2018020},
  bibsource    = {dblp computer science bibliography, https://dblp.org}
}

@inproceedings{paolo,
  author       = {Paolo Santini and
                  Giulia Rafaiani and
                  Massimo Battaglioni and
                  Franco Chiaraluce and
                  Marco Baldi},
  title        = {Optimization of a {R}eed-{S}olomon code-based protocol against blockchain
                  data availability attacks},
  booktitle    = {2022 {IEEE} International Conference on Communications Workshops,
                  {ICC} Workshops 2022, Seoul, Korea, May 16-20, 2022},
  pages        = {31--36},
  publisher    = {{IEEE}},
  year         = {2022},
  url          = {https://doi.org/10.1109/ICCWorkshops53468.2022.9814692},
  doi          = {10.1109/ICCWORKSHOPS53468.2022.9814692},
  bibsource    = {dblp computer science bibliography, https://dblp.org}
}

@ARTICLE{Cai2020,

  author={Cai, Han and Miao, Ying and Schwartz, Moshe and Tang, Xiaohu},

  journal={IEEE Transactions on Information Theory}, 

  title={On Optimal Locally Repairable Codes With Multiple Disjoint Repair Sets}, 

  year={2020},

  volume={66},

  number={4},

  pages={2402-2416},

  doi={10.1109/TIT.2019.2944397}}

@ARTICLE{Wang2014,

  author={Wang, Anyu and Zhang, Zhifang},

  journal={IEEE Transactions on Information Theory}, 

  title={Repair Locality With Multiple Erasure Tolerance}, 

  year={2014},

  volume={60},

  number={11},

  pages={6979-6987},

  doi={10.1109/TIT.2014.2351404}}

@ARTICLE{Tamo2014,

  author={Tamo, Itzhak and Barg, Alexander},

  journal={IEEE Transactions on Information Theory}, 

  title={A Family of Optimal Locally Recoverable Codes}, 

  year={2014},

  volume={60},

  number={8},

  pages={4661-4676},

  doi={10.1109/TIT.2014.2321280}}

@ARTICLE{Kamath2014,

  author={Kamath, Govinda M. and Prakash, N. and Lalitha, V. and Kumar, P. Vijay},

  journal={IEEE Transactions on Information Theory}, 

  title={Codes With Local Regeneration and Erasure Correction}, 

  year={2014},

  volume={60},

  number={8},

  pages={4637-4660},

  doi={10.1109/TIT.2014.2329872}}

@inproceedings{Gopalan2017,
  author       = {Parikshit Gopalan and
                  Guangda Hu and
                  Swastik Kopparty and
                  Shubhangi Saraf and
                  Carol Wang and
                  Sergey Yekhanin},
  editor       = {Philip N. Klein},
  title        = {Maximally Recoverable Codes for Grid-like Topologies},
  booktitle    = {Proceedings of the Twenty-Eighth Annual {ACM-SIAM} Symposium on Discrete
                  Algorithms, {SODA} 2017, Barcelona, Spain, Hotel Porta Fira, January
                  16-19},
  pages        = {2092--2108},
  publisher    = {{SIAM}},
  year         = {2017},
  doi          = {10.1137/1.9781611974782.136},
}

@inproceedings{NazirkhanovaNT22,
  author       = {Kamilla Nazirkhanova and
                  Joachim Neu and
                  David Tse},
  editor       = {Maurice Herlihy and
                  Neha Narula},
  title        = {Information Dispersal with Provable Retrievability for Rollups},
  booktitle    = {Proceedings of the 4th {ACM} Conference on Advances in Financial Technologies,
                  {AFT} 2022, Cambridge, MA, USA, September 19-21, 2022},
  pages        = {180--197},
  publisher    = {{ACM}},
  year         = {2022},
  url          = {https://doi.org/10.1145/3558535.3559778},
  doi          = {10.1145/3558535.3559778},
  bibsource    = {dblp computer science bibliography, https://dblp.org}
}

@article{TamoBF16,
  author       = {Itzhak Tamo and
                  Alexander Barg and
                  Alexey A. Frolov},
  title        = {Bounds on the Parameters of Locally Recoverable Codes},
  journal      = {{IEEE} Trans. Inf. Theory},
  volume       = {62},
  number       = {6},
  pages        = {3070--3083},
  year         = {2016},
  url          = {https://doi.org/10.1109/TIT.2016.2518663},
  doi          = {10.1109/TIT.2016.2518663},
  bibsource    = {dblp computer science bibliography, https://dblp.org}
}

@article{Brakensiek2025,
  author       = {Joshua Brakensiek and
                  Manik Dhar and
                  Sivakanth Gopi},
  title        = {Improved Constructions and Lower Bounds for Maximally Recoverable
                  Grid Codes},
  journal      = {CoRR},
  volume       = {abs/2509.15013},
  year         = {2025},
}

@article{Kong2021,
  title        = {New bounds on the field size for maximally recoverable codes instantiating grid-like topologies},
  author       = {Xiangliang Kong and Jingxue Ma and Gennian Ge},
  journal      = {Journal of Algebraic Combinatorics},
  volume       = {54},
  number       = {2},
  pages        = {529--557},
  year         = {2021},
  doi          = {10.1007/s10801-021-01013-1},
}

@book{LN,
  title     = {Finite Fields},
  author    = {Rudolf Lidl and Harald Niederreiter},
  edition   = {2},
  publisher = {Cambridge University Press},
  year      = {1997},
  series    = {Encyclopedia of Mathematics and Its Applications},
  volume    = {20},
  isbn      = {9780521392310},
  address   = {Cambridge, UK}
}

@book{roth,
  title     = {Introduction to Coding Theory},
  author    = {Ron M. Roth},
  publisher = {Cambridge University Press},
  year      = {2006},
  address   = {Cambridge, UK},
  isbn      = {9780521845045},
  pages     = {566},
}

@article{AlBassam2018,
  title        = {Fraud and Data Availability Proofs: Maximising Light Client Security and Scaling Blockchains with Dishonest Majorities},
  author       = {Mustafa Al-Bassam and Alberto Sonnino and Vitalik Buterin},
  journal      = {arXiv preprint},
  archivePrefix= {arXiv},
  eprint       = {1809.09044},
  year         = {2018},
  primaryClass = {cs.CR},
  url          = {https://arxiv.org/abs/1809.09044}
}

@article{forneyGMD,
  author =  {Forney, Jr., G. David},
  title   = {Generalized Minimum Distance Decoding},
  journal = {IEEE Transactions on Information Theory},
  volume  = {12},
  number  = {2},
  pages   = {125--131},
  year    = {1966},
  doi     = {10.1109/TIT.1966.1053874}
}

@INPROCEEDINGS{WZ14,
  author={Wang, Anyu and Zhang, Zhifang},
  booktitle={2014 IEEE International Symposium on Information Theory}, 
  title={Repair locality from a combinatorial perspective}, 
  year={2014},
  pages={1972-1976},
}

\end{document}